\pdfoutput=1
\documentclass[11pt]{article}

\let\originalleft\left
\let\originalright\right
\renewcommand{\left}{\mathopen{}\mathclose\bgroup\originalleft}
\renewcommand{\right}{\aftergroup\egroup\originalright}

\usepackage[colorlinks, allcolors=blue, linktocpage, breaklinks]{hyperref}
\usepackage{amsmath, amssymb, amsthm, amsfonts}
\usepackage[capitalize]{cleveref}
\usepackage{enumerate}
\usepackage{enumitem}
\usepackage{autobreak, bbm, geometry, mathtools, titling, xspace}

\usepackage[backend=biber, style=alphabetic, isbn=false, backref, maxbibnames=99]{biblatex}
\DefineBibliographyStrings{english}{
	backrefpage = {p\adddot},
	backrefpages = {pp\adddot}
}

\allowdisplaybreaks

\theoremstyle{plain}
\newtheorem{thm}{Theorem}[section]
\newtheorem{lem}[thm]{Lemma}
\newtheorem{cor}[thm]{Corollary}

\newtheorem{cnj}[thm]{Conjecture}
\newtheorem{fact}[thm]{Fact}

\crefname{fact}{Fact}{Facts}
\crefname{lem}{Lemma}{Lemmas}
\crefname{thm}{Theorem}{Theorems}

\theoremstyle{definition}
\newtheorem{dfn}[thm]{Definition}
\newtheorem{example}[thm]{Example}

\theoremstyle{remark}
\newtheorem*{remark*}{Remark}

\newcommand*{\bits}{\{0,1\}}
\newcommand*{\cart}{\mathbin\square}

\newcommand*{\Ex}{\mathbb{E}}
\newcommand*{\Ind}[1]{\mathbbm 1\!\{#1\}}
\newcommand*{\N}{\mathbb{N}}
\newcommand*{\R}{\mathbb{R}}
\newcommand*{\Rnng}{\R_{\ge0}}
\DeclarePairedDelimiter{\ceil}{\lceil}{\rceil}
\DeclarePairedDelimiter{\floor}{\lfloor}{\rfloor}

\newcommand*{\aczero}{$\mathrm{AC}^0$\xspace}
\newcommand*{\Bu}{\beta_{\mathrm o}}
\newcommand*{\colt}[2]{\mathrm{Sub}_{#2}(#1)}
\newcommand*{\const}{\mathrm{const}}
\newcommand*{\dst}{\Delta^*}
\newcommand*{\Du}{\Delta_{\mathrm o}}
\newcommand*{\ER}[2]{\mathbf{ER}\left(#1,#2\right)}
\newcommand*{\ham}[2]{K_{#1}^{#2}}
\newcommand*{\kcol}{\kappa}
\newcommand*{\kmarx}{\mathit{emb}}
\newcommand*{\mblowup}[2]{#1^{\left(#2\right)}}
\newcommand*{\nul}{\mathrm{null}}
\newcommand*{\rg}[3]{\mathbf{#1}_{#2, #3}}
\newcommand*{\seq}{\mathrm{Seq}}
\newcommand*{\subcol}[1][G]{\text{$#1$-$\mathsf{SUB}$}}
\newcommand*{\subuncol}[1][G]{\text{$#1$-$\mathsf{SUB}_{\mathrm{uncol}}$}}
\newcommand*{\tcl}{\theta}
\newcommand*{\tw}{\mathit{tw}}

\author{Gregory Rosenthal\thanks{Email: \href{mailto:rosenthal@cs.toronto.edu}{\color{black}\texttt{rosenthal@cs.toronto.edu}}. Supported by NSERC (PGS D).}\\University of Toronto}
\title{Beating Treewidth for Average-Case Subgraph Isomorphism}
\date{\vspace{-5ex}}
\begin{document}
\begin{NoHyper}
\maketitle
\end{NoHyper}
\begin{abstract}
	For any fixed graph $G$, the subgraph isomorphism problem asks whether an $n$-vertex input graph has a subgraph isomorphic to $G$. A well-known algorithm of Alon, Yuster and Zwick (1995) efficiently reduces this to the ``colored" version of the problem, denoted $\subcol$, and then solves $\subcol$ in time $O(n^{\tw(G)+1})$ where $\tw(G)$ is the treewidth of $G$. Marx (2010) conjectured that $\subcol$ requires time $\Omega(n^{\const\cdot\tw(G)})$ and, assuming the Exponential Time Hypothesis, proved a lower bound of $\Omega(n^{\const\cdot\kmarx(G)})$ for a certain graph parameter $\kmarx(G) \ge \Omega(\tw(G)/\log \tw(G))$. With respect to the size of \aczero circuits solving $\subcol$ in the average case, Li, Razborov and Rossman (2017) proved (unconditional) upper and lower bounds of $O(n^{2\kcol(G)+\const})$ and $\Omega(n^{\kcol(G)})$ for a different graph parameter $\kcol(G) \ge \Omega(\tw(G)/\log \tw(G))$.
	
	Our contributions are as follows. First, we prove that $\kmarx(G)$ is $O(\kcol(G))$ for all graphs $G$. Next, we show that $\kcol(G)$ can be asymptotically less than $\tw(G)$; for example, if $G$ is a hypercube then $\kcol(G)$ is $\Theta\left(\tw(G)\big/\sqrt{\log \tw(G)}\right)$. This implies that the average-case complexity of $\subcol$ is $n^{o(\tw(G))}$ when $G$ is a hypercube. Finally, we construct \aczero circuits of size $O(n^{\kcol(G)+\const})$ that solve $\subcol$ in the average case, closing the gap between the upper and lower bounds of Li et al.
\end{abstract}
\section{Introduction}\label{intro}

The subgraph isomorphism problem asks, given graphs $X$ and $G$, whether $X$ has a subgraph isomorphic to $G$. In the ``colored" or ``partitioned" version of the problem, each vertex of the larger graph $X$ comes with a ``color" from the vertex set of $G$, and we ask whether $X$ has a subgraph that is isomorphic to $G$ with respect to this coloring. We denote the uncolored and colored subgraph isomorphism problems by $\subuncol(X)$ and $\subcol(X)$ respectively.

Subgraph isomorphism is NP-complete (e.g.\ if $G$ is a clique or Hamiltonian cycle), so research has focused on algorithms for a variety of special cases in the context of parameterized complexity, surveyed in \cite{MP14}. If $G$ is a fixed graph on $k$ vertices then $\subuncol$ is solvable in time $O(n^k)$ by brute force, where (here and throughout this section) $n$ is the order of the input graph. The color-coding algorithm of Alon, Yuster and Zwick~\cite{AYZ95} improves on this by efficiently reducing $\subuncol$ to $\subcol$ and solving the latter in time $O(n^{\tw(G)+1})$, where $\tw(G)$ is the treewidth of the fixed graph $G$.

The exponent $\tw(G)+1$ can sometimes be improved using fast matrix multiplication~\cite{NP85,EG04}, but no significantly faster algorithm is known for either the colored or uncolored subgraph isomorphism problem. Marx~\cite{Mar10} conjectured the following:

\begin{cnj}\label{conj-marx}
	There is no class $\mathcal G$ of graphs with unbounded treewidth, no algorithm $\mathbb A$ that on inputs $G$ and $X$ solves $\subcol(X)$, and no function $f$ such that if $G$ is in $\mathcal G$ then $\mathbb A$ runs in time $f(G) n^{o(\tw(G))}$.
\end{cnj}

Marx~\cite{Mar10} came close to proving \cref{conj-marx} assuming the Exponential Time Hypothesis (ETH)~\cite{IPZ01}, which is the hypothesis that solving 3SAT on $n$ variables requires $2^{\Omega(n)}$ time. We state his result in terms of a parameter $\kmarx(G)$ (short for ``embedding") which we will define in \cref{sec-3}:

\begin{thm}[{\cite{Mar10}}]\label{mar10main}
	Assuming ETH, there is no class $\mathcal G$ of graphs with unbounded treewidth, no algorithm $\mathbb A$ that on inputs $G$ and $X$ solves $\subcol(X)$, and no function $f$ such that if $G$ is in $\mathcal G$ then $\mathbb A$ runs in time $f(G) n^{o(\kmarx(G))}$.
\end{thm}

Marx~\cite{Mar10} proved that $\kmarx(G)$ is $\Omega(\tw(G)/\log \tw(G))$, so \cref{mar10main} comes within a logarithmic factor in the exponent of proving \cref{conj-marx} (under ETH). However, our results include a counterexample to an average-case analogue of \cref{conj-marx}, in a sense that will be made precise in \cref{sec-2a}. Moreover, this counterexample holds in \aczero, i.e.\ on unbounded-fanin boolean circuits of depth depending only on $G$.

Li, Razborov and Rossman~\cite{LRR17} proved that for fixed $G$, the average-case \aczero complexity of $\subcol$ is between $n^{\kcol(G)-o(1)}$ and $n^{2\kcol(G)+c}$, where $\kcol(G)$ is a graph property and $c$ is an absolute constant.\footnote{In \cite{LRR17}, the parameter $\kcol(G)$ was called $\kappa_{\mathrm{col}}(G)$.} (See \cref{sec-2a} for Li et al.'s definition of $\kcol(G)$; we also prove that $\kcol(G)$ can be equivalently defined in terms of the transition matrix of a certain random walk on $G$.) We tighten this gap, answering a question posed in \cite{LRR17}:

\begin{thm}\label{asdf-aczero}
	There is a constant $c>0$ such that for any fixed graph $G$, the average-case \aczero complexity of $\subcol$ is at most $n^{\kcol(G)+c}$.
\end{thm}

We observe that a similar result holds easily on Turing machines, using as a subroutine the \emph{sort-merge join} algorithm from relational algebra. This involves sorting, which cannot be done in (polynomial-size) \aczero~\cite{Has86}, so our circuit instead uses hashing that relies on concentration of measure for subgraphs of random graphs.

Li et al.~{\cite{LRR17}} also proved that $\kcol(G)$ is between $\Omega(\tw(G)/\log \tw(G))$ and $\tw(G)+1$, from which it follows that the \emph{worst}-case complexity of $\subcol$ is at least $n^{\Omega(\tw(G)/\log \tw(G))}$ in \aczero. Li et al.\ posed the question of whether $\kcol(G)$ is $\Theta(\tw(G))$; an affirmative answer would have implied that \cref{conj-marx} holds in \aczero.

However, the following example separates $\kcol$ from treewidth. The Hamming graph $\ham q d$ has vertex set $\{1,\dotsc,q\}^d$ and edges between every two vertices that differ in exactly one coordinate. It is already known that $\ham q d$ has treewidth $\Theta\left(q^d\big/\sqrt d\right)$~\cite{CK06}. We prove the following:
\begin{thm}\label{main}
	$\kcol\left(\ham q d\right)$ is $\Theta(q^d/d)$.
\end{thm}

Thus, if $G$ is the hypercube graph $\ham 2 d$ for example, then $\kcol(G)$ is $\Theta\left(\tw(G)\big/\sqrt{\log \tw(G)}\right)$. It follows that an average-case analogue of \cref{conj-marx} is false if $\mathcal G$ is taken to be the set of all hypercubes. We also prove the following (for arbitrary graphs $G$):
\begin{thm}\label{secondary}
	$\kmarx(G)$ is $O(\kcol(G))$.
\end{thm}

Because of \cref{secondary}, even if our upper bound generalizes to the worst case, it is still consistent with current knowledge (in particular \cref{mar10main}) that ETH is true. Another consequence of \cref{secondary} is that the lower bound from \cref{mar10main} holds unconditionally in \aczero.

It follows from \cref{main,secondary} that if $G$ is a hypercube then $\kmarx(G) \le O(\kcol(G)) \le o(\tw(G))$, so proving that \cref{conj-marx} holds under ETH cannot be done by proving that $\kmarx(G)$ is $\Theta(\tw(G))$. In fact, this conclusion was already known: Alon and Marx~\cite{AM11} proved that if $G$ is a 3-regular expander then $\kmarx(G)$ is $\Theta(\tw(G)/\log \tw(G))$. Li et al.~\cite{LRR17} proved that if $G$ is a 3-regular expander then $\kcol(G)$ is $\Theta(\tw(G))$, which makes our separation of $\kcol$ from treewidth more surprising. On the other hand, we will see that \cref{secondary} is tight in the case of Hamming graphs.

We can make a similar statement regarding \aczero. Amano~\cite{Ama10} observed that the color-coding algorithm for $\subcol$ can be implemented by \aczero circuits of size $O(n^{\tw(G)+1})$ for fixed $G$. Our separation of $\kcol$ from treewidth implies that if \cref{conj-marx} holds in \aczero, then this cannot be proved using average-case complexity as defined here and in \cite{LRR17}.

The paper is organized as follows. In \cref{sec-2} we introduce some notation and definitions. In \cref{sec-2a} we define the average-case problem and $\kcol(G)$, and give an $\tilde O(n^{\kcol(G)})$-time algorithm for the average-case problem. In \cref{sec-3} we define $\kmarx(G)$ and prove that $\kmarx(G)$ is $O(\kcol(G))$. In \cref{sec-4} we prove that $\kcol\left(\ham q d\right)$ is $\Theta(q^d/d)$, and obtain as a corollary that $\kmarx\left(\ham q d\right)$ is $\Theta(q^d/d)$ as well. We also summarize the proof of Chandran and Kavitha~\cite{CK06} that $\tw\left(\ham q d\right)$ is $\Theta\left(q^d\big/\sqrt d\right)$. In \cref{sec-5} we prove our \aczero upper bound.

\section{Preliminaries}\label{sec-2}
It will be convenient to define $\tilde O(f(n)) = f(n)\log^{O(1)}n$. (This differs from the standard notation when $f(n) = n^{o(1)}$.) We will often fix a graph $G$, in which case the constants hidden in asymptotic notation are allowed to depend on $G$.

We use \textbf{boldface} to denote random variables. The indicator variable $\Ind E$ equals 1 if the event $E$ occurs and 0 otherwise. Expected value is denoted $\Ex[\cdot]$. An event occurs \emph{asymptotically almost surely (a.a.s.)}\ if it occurs with probability $1-o(1)$ as $n$ goes to infinity.

Let $[k] = \{1, \dotsc, k\}$ for $k\in\N$. If a positive real number $x$ is used in a context where a natural number is expected (for example $[x]$), it's because $x$ can be rounded arbitrarily to $\ceil x$ or $\floor x$ without affecting the asymptotic behavior of whatever is being considered.
\subsection{Graphs}\label{graphs-intro}

All graphs we consider are simple and undirected, and may have isolated vertices. If $G$ is a graph then let $V(G)$ and $E(G)$ denote its vertex and edge sets, with respective cardinalities $v(G)$ and $e(G)$. If $u$ and $v$ are adjacent vertices then we denote the edge connecting them by $uv$ or $vu$. A graph $H$ is a subgraph of $G$, denoted $H \subseteq G$, if $V(H) \subseteq V(G)$ and $E(H) \subseteq E(G)$.

\begin{dfn}[Colored subgraph isomorphism problem]
	For graphs $G$ and $X$, where $X$ comes with a coloring $\chi: V(X) \rightarrow V(G)$, the problem $\subcol(X)$ asks whether $X$ has a subgraph $G^\prime$ such that $\chi$ (restricted to $V(G^\prime)$) is an isomorphism from $G^\prime$ to $G$.
\end{dfn}	

For $U \subseteq V(G)$ let $G[U]$ be the induced subgraph of $G$ on $U$, and more generally let $G[U_1, \dotsc, U_k] = G[U_1 \cup \dotsb \cup U_k]$. Let $G-U = G[V(G)-U]$, and for $H\subseteq G$ let $G-H = G-V(H)$.

When the parent graph $G$ is clear in context, let $\deg(u)$ be the degree of a vertex $u$, and for disjoint $S,T\subseteq V(G)$ let $e(S,T)$ be the number of edges between $S$ and $T$. Similarly, for vertex-disjoint graphs $A,B\subseteq G$ let $e(A,B) = e(V(A), V(B))$.

Let $G \cap H$ be the graph with vertex set $V(G) \cap V(H)$ and edge set $E(G) \cap E(H)$, and define $G \cup H$ similarly. Note that $G\cap H$ may have isolated vertices even if $G$ and $H$ do not. If $A\subseteq B$ are graphs then let $[A,B] = \{H\mid A\subseteq H\subseteq B\}$, and let $(A,B]$ be the same interval without $A$, etc.

The Cartesian product of graphs $G$ and $H$, denoted $G \cart H$, has vertex set $V(G) \times V(H)$ and edges $(u,v_1)(u,v_2)$ for all $u\in V(G)$ and $v_1v_2\in E(H)$, and $(u_1,v)(u_2,v)$ for all $u_1u_2\in E(G)$ and $v\in V(H)$. Let $G^d$ be the Cartesian product of $d$ copies of $G$.

We denote by $K_k$ the complete graph on $k$ vertices, also called the $k$-clique. It follows that $\ham q d$ has vertex set $[q]^d$, and two vertices are adjacent if and only if they differ in exactly one coordinate. Such graphs are called \emph{Hamming graphs}. A special case is the $d$-dimensional hypercube $Q_d = \ham 2 d$; we will use $\bits^d$ for its vertex set.

\begin{dfn}[Graph minor]
	A graph $H$ is a \emph{minor} of a graph $G$ if there exists a \emph{minor mapping} $\phi$ assigning a connected component of $G$ to each vertex of $H$, such that $\phi(u)$ and $\phi(v)$ are vertex-disjoint for all $u\neq v$, and if $uv\in E(H)$ then there exists an edge in $G$ with endpoints in $\phi(u)$ and $\phi(v)$.
\end{dfn}

In particular, any subgraph of $G$ is also a minor of $G$ (e.g.\ let $\phi$ be the identity).

\begin{dfn}[Treewidth]
	A \emph{tree decomposition} of a graph $G$ is a tree $T$ whose vertices are subsets of $V(G)$ (called ``bags"), such that each vertex and edge of $G$ is contained in at least one of the bags, and for all $u \in V(G)$, the induced subgraph of $T$ on the bags that contain $u$ is a connected subtree of $T$. The \emph{width} of $T$ is one less than the size of the smallest bag, and the \emph{treewidth} of $G$, denoted $\tw(G)$, is the minimum width over all tree decompositions.
\end{dfn}

Roughly speaking, a graph has small treewidth if and only if it's ``similar to a tree". See e.g.\ \cite{Bod98, BK08} for further background, and \cite{HW17} for a survey of parameters that are polynomially tied to treewidth.

The edge expansion of a graph $G$ is defined as follows:
\begin{equation*}
h(G) = \min_{\emptyset\subset U\subset V(G)}\frac{e(U,V(G)-U)}{\min(|U|,|V(G)-U|)}.
\end{equation*}
A bounded-degree expander is a graph with edge expansion $\Omega(1)$ and maximum degree $O(1)$ (see \cite{HLW06} for a survey). Let $\lambda_i(G)$ be the $i$'th largest eigenvalue of the adjacency matrix of $G$. We will use the following half of Cheeger's Inequality:
\begin{fact}[\cite{AM85}]\label{cheeger}
	If $G$ is a $d$-regular graph then $h(G) \ge (d-\lambda_2(G))/2$.
\end{fact}

Finally, let $\ER n p$ be the Erd\H{o}s-R\'{e}nyi graph on $n$ vertices in which each possible edge exists independently with probability $p$.

\section{The Average-Case Problem and the Parameter \texorpdfstring{$\kcol(G)$}{kappa(G)}}\label{sec-2a}
\subsection{Threshold Random Graphs}\label{trg-intro}

First we will define \emph{threshold weightings}, which assign weights to the vertices and edges of a graph subject to certain constraints. Then we will define a family of random graphs for each threshold weighting. The content in this subsection is essentially all from \cite{LRR17}.

\begin{dfn}\label{tp}
	A threshold weighting on a graph $G$ is a pair $(\alpha,\beta) \in [0,1]^{V(G)} \times [0,2]^{E(G)}$ with the following property. For $H\subseteq G$ let $\alpha(H) = \sum_{u\in V(H)}\alpha(u)$ and $\beta(H) = \sum_{e\in E(H)}\beta(e)$, and let $\Delta(H) = \alpha(H)-\beta(H)$. Then, $\Delta(H)\ge0$ for all $H\subseteq G$, and $\Delta(G)=0$. Let $\tcl(G)$ be the set of threshold weightings on $G$.
	
	We will often denote $\Delta = (\alpha,\beta)$ in a slight abuse of notation. (Since $\Delta(u)=\alpha(u)$ if $u$ is a single vertex, the pair $(\alpha,\beta)$ is uniquely determined by $\Delta$.) The requirement that $\alpha$ be nonnegative is redundant because it's a special case of the requirement that $\Delta$ be nonnegative. The requirement that $\beta \le 2$ is also redundant because for every edge $uv$,
	\begin{equation*}
	0 \le \Delta(uv) = \alpha(u) + \alpha(v) - \beta(uv) \le 2 - \beta(uv).
	\end{equation*}
	It will sometimes be convenient to define $\beta(e)=0$ for $e\notin E(G)$, e.g.\ for disjoint sets $S, T \subseteq V(G)$ let $\beta(S, T) = \sum_{u\in S, v \in T}\beta(uv)$, and for vertex-disjoint $A, B \subseteq G$ let $\beta(A, B) = \beta(V(A), V(B))$.
\end{dfn}

\begin{example}[Markov Chains]\label{mc-example}
	Let $M \in \Rnng^{V(G)\times V(G)}$ be a column stochastic matrix (meaning each column sums to 1) such that if $M_{u,v} \neq 0$ then either $u = v$ or $uv \in E(G)$. Let $\alpha(u) = 1-M_{u,u}$ for all $u$, and $\beta(uv) = M_{u,v} + M_{v,u}$ for all $u \neq v$. Then for all $H\subseteq G$,
	\begin{equation}\label{mc-eq}
	\Delta(H) = \sum_{\mathclap{\substack{v\in V(H) \\ uv\in E(G)-E(H)}}}M_{u,v} \ge 0,
	\end{equation}
	with equality if $H=G$. In fact, we prove that every threshold weighting is equivalent to at least one Markov Chain (\cref{markov}).
\end{example}

The following threshold weighting will be especially important, and can be thought of as representing a
uniform random walk on $G$:

\begin{dfn}\label{mc-unif}
	If $G$ lacks isolated vertices then let $\Du = (1,\Bu) \in \tcl(G)$ be the threshold weighting generated in \cref{mc-example} when $M_{u,v} = \Ind{uv\in E(G)}/\deg(v)$. That is, $\Du = (\alpha,\beta)$, where $\alpha(u)=1$ for all $u$ and $\beta(uv) = 1/\deg(u) + 1/\deg(v)$ for all $uv$. If $G$ is $d$-regular then this simplifies to $\Du = (1,\Bu) = (1,2/d)$.
\end{dfn}

Now we define threshold random graphs:

\begin{dfn}
	For $\Delta = (\alpha,\beta) \in \tcl(G)$ let $\rg X \Delta n$ be the graph with vertices $u_i$ for $u\in V(G)$ and $i\in[n^{\alpha(u)}]$, and for $uv \in E(G)$, each edge $u_i v_j$ independently with probability $n^{-\beta(uv)}$. The graph $\rg X \Delta n$ comes with the coloring to $G$ defined by $u_i \mapsto u$.
	
	For $H\subseteq G$ and $X$ in the support of $\rg X \Delta n$, let $\colt H X$ be the set of subgraphs $H^\prime \subseteq X$ such that the aforementioned coloring (restricted to $V(H^\prime)$) is an isomorphism from $H^\prime$ to $H$. We say that such a graph $H^\prime$ is ``$H$-colored". Note that $\colt H X$ can be identified with a subset of $\prod_{u \in V(H)}[n^{\alpha(u)}]$.
\end{dfn}

\begin{lem}\label{ex-sub-count}
	If $\Delta \in \tcl(G)$ and $H\subseteq G$ then $\Ex[|\colt H {\rg X \Delta n}|] = n^{\Delta(H)}(1 \pm o(1))$.
\end{lem}
\begin{proof}
	Let $(\alpha,\beta) = \Delta$. The set $\colt H {\rg X \Delta n}$ contains each of its $n^{\alpha(H)}$ possible elements with probability $n^{-\beta(H)}$, so the result follows from linearity of expectation. (The $1 \pm o(1)$ accounts for having to round $n^{\alpha(\cdot)}$ to an integer.)
\end{proof}

\cref{ex-sub-count} motivates the requirements that $\Delta$ be nonnegative everywhere and that $\Delta(G)=0$. Recall that the problem $\subcol(X)$ asks whether $\colt G X$ is the empty set. Since $\Delta(G)$ is required to be zero, it follows that $\colt G {\rg X \Delta n}$ has (approximately) one element on average, and the probability that $\colt G {\rg X \Delta n}$ is empty is known to be bounded away from 0 and 1 as $n$ goes to infinity~\cite{LRR17}.

\subsection{The Parameter \texorpdfstring{$\kcol(G)$}{kappa(G)} and an Algorithm for the Average Case}

We now define $\kcol(G)$:

\begin{dfn}[\cite{LRR17}]
	Let $G$ be a graph with no isolated vertices. Let $\seq(G)$ be the set of \emph{union sequences}, meaning sequences $(H_1, \dotsc, H_k)$ of distinct subgraphs of $G$ such that $H_k = G$ and each $H_i$ is either an edge or the union of two previous graphs in the sequence. For $\Delta \in \tcl(G)$ let $\kappa_{\Delta}(G) = \min_{S\in\seq(G)}\max_{H\in S}\Delta(H)$. Finally, let $\kcol(G) = \max_{\Delta\in\tcl(G)} \kappa_{\Delta}(G)$.
\end{dfn}

To simplify the exposition, whenever we refer to $\kcol(G)$, the graph $G$ is implicitly assumed to lack isolated vertices. Li et al.~\cite{LRR17} proved that for any fixed $G$, \aczero circuits solving $\subcol(\rg X \Delta n)$ a.a.s.\ require size at least $n^{\kappa_\Delta(G)-o(1)}$ and at most $n^{2\kappa_\Delta(G)+c}$ (where $c$ is an absolute constant). The results about average-case complexity described in \cref{intro} are with respect to a $\Delta$ such that $\kappa_\Delta(G) = \kcol(G)$.

\begin{thm}\label{alg-main}
	The problem $\subcol(\rg X \Delta n)$ can be solved in time $\tilde O(n^{\kappa_\Delta(G)}) \le \tilde O(n^{\kcol(G)})$ a.a.s.\ for any fixed $G$.
\end{thm}
\begin{proof}
	First we prove a weaker upper bound of $\tilde O(n^{2\kappa_\Delta(G)})$, in a manner analogous to the circuit from \cite{LRR17}, and then we describe a modification (on Turing machines) that removes the factor of 2 from the exponent. Later we will remove the factor of 2 in \aczero using a different approach, summarized at the beginning of \cref{sec-5}.
	
	Let $S$ be a union sequence such that $\kappa_\Delta(G) = \max_{H \in S}\Delta(H)$. For any $H \in S$, by \cref{ex-sub-count} and Markov's Inequality, $P\left(|\colt H {\rg X \Delta n}| > n^{\Delta(H)}\log n\right) \le 1/\log n$. (We will obtain a tighter bound of $P(|\colt H {\rg X \Delta n}| > \tilde O(n^{\Delta(H)})) \le n^{-\omega(1)}$ in \cref{good-subsection}.) By a union bound it follows that if $X \sim \rg X \Delta n$ then $\max_{H \in S}|\colt H X| \le \tilde O(n^{\kappa_\Delta(G)})$ a.a.s. Assume this condition holds for $X$. For each successive $H$ in $S$, compute $\colt H X$ as follows. If $H$ is a single edge then this is trivial. Otherwise $H = A \cup B$ for some previous $A, B \in S$, in which case $\colt H X$ is the set of $\mathcal A \cup \mathcal B$ such that $\mathcal A \in \colt A X, \mathcal B \in \colt B X$ and the projections of $\mathcal A$ and $\mathcal B$ onto $[n]^{V(A\cap B)}$ are equal. Therefore $\colt H X$ can be computed by brute force in time $\tilde O(|\colt A X| \cdot |\colt B X|) \le \tilde O(n^{2\kappa_\Delta(G)})$. Finally, check whether $\colt G X$ is empty.
	
	We can save a quadratic factor by computing $\colt H X$ from $\colt A X$ and $\colt B X$ as follows. (This is a case of the \emph{sort-merge join} algorithm for computing the \emph{natural join} of two relations, as defined in database theory~\cite{SKS11}.) Define a partial order on $[n]^{V(A)} \cup [n]^{V(B)}$ by projecting onto $[n]^{V(A\cap B)}$ and applying the lexicographic order on $[n]^{V(A\cap B)}$. Sort $\colt A X$ and $\colt B X$ in nondecreasing order, and for convenience add the symbol $\perp$ to the end of both sorted lists. Let $\mathcal A$ and $\mathcal B$ be the first elements of $\colt A X$ and $\colt B X$ respectively, and initialize an empty accumulator (which will ultimately equal $\colt H X$). While $\mathcal A\neq\perp$ and $\mathcal B\neq\perp$, do the following. If $\mathcal A<\mathcal B$ then let $\mathcal A$ be the next element of $\colt A X$. If $\mathcal B<\mathcal A$ then let $\mathcal B$ be the next element of $\colt B X$. Otherwise, let $\mathcal B^\prime = \mathcal B$, and while $\mathcal B^\prime \neq \perp$ and the projections of $\mathcal A$ and $\mathcal B^\prime$ onto $[n]^{V(A\cap B)}$ are equal, add $\mathcal A \cup \mathcal B^\prime$ to the accumulator and let $\mathcal B^\prime$ be the next element of $\colt B X$. Then (once the procedure involving $\mathcal B^\prime$ has finished) let $\mathcal A$ be the next element of $\colt A X$.
	
	Sorting $\colt A X$ and $\colt B X$ takes $\tilde O(|\colt A X| + |\colt B X|)$ comparisons, each of which takes $\tilde O(1)$ time, and then computing $\colt H X$ takes $\tilde O(|\colt A X| + |\colt B X| + |\colt H X|) \le \tilde O(n^{\kappa_\Delta(G)})$ time.
\end{proof}

We will use the following graph-theoretic properties of $\kcol(G)$:
\begin{thm}[\cite{LRR17}\footnote{Specifically, Corollary 4.2, Theorem 4.9, and Theorem 5.1 of \cite{LRR17} correspond to \cref{max-at-one,kcol-expansion,kcol-mm} respectively.}]
	Let $G$ be a graph with no isolated vertices.
	\begin{enumerate}[label=\textnormal{(\roman*)}, ref={\thethm(\roman*)}]
		\item There exists $\Delta = (1,\beta) \in \tcl(G)$ (meaning $\Delta(u)=1$ for all vertices $u$) such that $\kcol(G) = \kappa_\Delta(G)$.\label[thm]{max-at-one}
		\item $\kcol(G) \ge v(G)h(G)/(3\max_{u\in V(G)}\deg(u))$, where $h(G)$ is the edge expansion of $G$. \label[thm]{kcol-expansion}
		\item If $G$ is a minor of some graph $H$ then $\kcol(G) \le \kcol(H)$. \label[thm]{kcol-mm}
	\end{enumerate}
\end{thm}

\begin{cor}\label{kcol-ex-cor}
	\leavevmode
	\begin{enumerate}[label=\textnormal{(\roman*)}, ref={\thecor(\roman*)}]
		\item If $G$ is a bounded-degree expander then $\kappa(G)$ is $\Omega(v(G))$. \label[cor]{kap-ex}
		\item If $G$ is a $d$-regular graph then $\kcol(G) \ge v(G)(1-\lambda_2(G)/d)/6$. \label[cor]{kcol-cheeger}
	\end{enumerate}
\end{cor}
\begin{proof}[Proof of \cref{kcol-ex-cor}]
	\cref{kap-ex} follows from \cref{kcol-expansion}, as observed by Li et al.~\cite{LRR17}. \cref{kcol-cheeger} follows from \cref{kcol-expansion,cheeger}.
\end{proof}

\section{The Parameter \texorpdfstring{$\kmarx(G)$}{emb(G)} and Proof that \texorpdfstring{$\kmarx(G)$}{emb(G)} is \texorpdfstring{$O(\kcol(G))$}{O(kappa(G))}}\label{sec-3}

Recall that $\kmarx(G)$ is significant because of its role in Marx's ETH-hardness result for $\subcol$, namely \cref{mar10main}.

\begin{dfn}[$\kmarx(G)$]\label{emb-def}
	Let $\mblowup G q$ be the graph formed by replacing each vertex of $G$ with a $q$-clique, i.e.\ it has vertices $u_i$ for all $u\in V(G)$ and $i\in[q]$, and edges $u_i v_j$ for all $u_i \neq v_j$ such that either $u=v$ or $uv \in E(G)$. Let $\kmarx(G)$ be the supremum of all $r>0$ for which there exists $m_0 = m_0(G,r)$ such that if $H$ is any graph with $m \ge m_0$ edges and no isolated vertices, then $H$ is a minor of $\mblowup G {\ceil{m/r}}$, and furthermore a minor mapping from $H$ to $\mblowup G {\ceil{m/r}}$ can be computed in time $f(G)m^{O(1)}$ for some function $f$. 
\end{dfn}

Although the requirement that such a minor mapping be efficiently computable is crucial in \cref{mar10main}, none of the other results about $\kmarx(G)$ that we reference or derive depend on this requirement, so we may safely ignore it going forward. The following example illustrates \cref{emb-def}:
\begin{example}[$\kmarx(K_k)$~\cite{Mar10}]
	Since $\mblowup {K_k} {\ceil{m/r}} = K_{k\ceil{m/r}}$, any graph $H$ with $m$ edges is a minor of $\mblowup {K_k} {\ceil{m/r}}$ if and only if $v(H) \le k\ceil{m/r}$. If $H$ has no isolated vertices then $H$ could have up to $2m$ vertices, so $2m \le k\ceil{m/r}$. Therefore $\kmarx(K_k) = k/2$: it is sufficient for $2m$ to be at most $km/r$ (i.e.\ $r \le k/2$), and no $r > k/2$ satisfies $2m \le k\ceil{m/r}$ for arbitrarily large $m$.
\end{example}

\begin{remark*}
	The name $\kmarx(G)$ comes from the fact that Marx~\cite{Mar10} called a minor mapping from $H$ to $\mblowup G q$ an ``embedding of depth $q$" from $H$ into $G$. Marx used the notation $\mblowup G q$, but the parameter $\kmarx(G)$ is new in the current paper, all results about $\kmarx(G)$ in \cite{Mar10,AM11} having been stated in terms of embeddings of some depth.
\end{remark*}

The following is used in our proof that $\kmarx(G)$ is $O(\kcol(G))$:
\begin{lem}\label{kcol-blowup}
	$\kcol\left(\mblowup G q\right) \le q\max(\kcol(G),2)$.
\end{lem}
\begin{proof}
	Let $\Delta=(\alpha,\beta) \in \tcl\left(\mblowup G q\right)$ such that $\kcol\left(\mblowup G q\right) = \kappa_\Delta\left(\mblowup G q\right)$. Define a threshold weighting $\Delta^\prime = (\alpha^\prime, \beta^\prime) \in \tcl(G)$ as follows: For all $u \in V(G)$ and $uv \in E(G)$,
	\begin{align*}
	\alpha^\prime(u) &= \frac{\Delta\left(\mblowup u q\right)}q = \frac1q\left(\sum_{i=1}^q \alpha(u_i) - \sum_{\mathclap{1 \le i < j \le q}} \beta(u_i u_j)\right),\\
	\beta^\prime(uv) &= \frac1q\sum_{i,j=1}^q \beta(u_i v_j).
	\end{align*}
	This is a threshold weighting because if $H\subseteq G$ then $\Delta^\prime(H) = \Delta(\mblowup H q)/q \ge 0$, with equality if $H = G$. It's also normalized to $\alpha^\prime \le 1$.
	
	Let $S^\prime$ be an optimal union sequence for $G$ with respect to $\Delta^\prime$. Construct a union sequence $S$ for $\mblowup G q$ as follows:
	\begin{enumerate}
		\item For each $e\in E(G)$ append an arbitrary union sequence for $\mblowup e q$.
		\item For each $H\in S^\prime$ (in order) append $\mblowup H q$.
	\end{enumerate}
	If $H\subseteq \mblowup e q$ then $\Delta(H) \le \alpha(\mblowup e q) \le 2q$, and we've already seen that $\Delta(\mblowup H q) = q\Delta^\prime(H)$ for all $H\in S^\prime$. Therefore,
	\begin{align*}
	\kcol\left(\mblowup G q\right) &= \kappa_\Delta\left(\mblowup G q\right) \le \max_{H\in S}\Delta(H) \le \max\left(2q, \max_{H\in S^\prime}q\Delta^\prime(H)\right) 
	= q\max(\kappa_{\Delta^\prime}(G),2) \\
	&\le q\max(\kcol(G),2).\qedhere
	\end{align*}
\end{proof}

Now we prove that $\kmarx(G)$ is $O(\kcol(G))$ (\cref{secondary}), using an argument similar to the proof by Marx~\cite{Mar10} that $\kmarx(G)$ is $O(\tw(G))$:
\begin{proof}
	Let $r>0$, and assume there exists an arbitrarily large 3-regular expander $H$ that's a minor of $\mblowup G {\ceil{e(H)/r}}$. Then by \cref{kap-ex,kcol-mm,kcol-blowup},
	\begin{equation*}
	e(H) = \Theta(v(H)) \le O(\kcol(H)) \le O\left(\kcol\left(\mblowup G {\ceil{e(H)/r}}\right)\right) \le O\left(\kcol(G)e(H)/r\right),
	\end{equation*}
	so $r$ must be $O(\kcol(G))$.
\end{proof}

Li et al.~\cite{LRR17} posed the question of whether \cref{mar10main} holds with $\kcol(G)$ in place of $\kmarx(G)$. By \cref{secondary} this would be a stronger bound, which makes the question even more interesting. This problem is open even in the case of 3-regular expanders: recall from \cref{intro} that if $G$ is a 3-regular expander then $\kmarx(G)$ is $\Theta(\tw(G)/\log\tw(G))$ and $\kcol(G)$ is $\Theta(\tw(G))$~\cite{AM11,LRR17}.

The fact that $\kcol(G)$ is $\Omega(\kmarx(G))$ gives an alternate proof, besides the one in \cite{LRR17}, that $\kcol(G)$ is $\Omega(\tw(G)/\log\tw(G))$.

\section{Separating \texorpdfstring{$\kcol$}{kappa} from Treewidth}\label{sec-4}
In \cref{kcol-clique} we prove that $\kcol(K_k) = k/4 + O(1)$, which is a special case of the more general result that $\kcol\left(\ham q d\right) = \Theta(q^d/d)$, and improves on the observation of Li et al.~\cite{LRR17} that $\kcol(K_k)$ is $\Theta(k)$. We obtain tighter multiplicative constants in the case $d=1$, and it provides an opportunity to illustrate the main ideas of our proof in a simpler setting, but it may be skipped without penalty. In \cref{o-qdd-even} we prove that $\kcol\left(\ham q d\right)$ is $O(q^d/d)$ when $q$ is even, which is sufficient to separate $\kcol$ from treewidth. Again, this case is cleaner than the general case and conveys most of the intuition behind it. In \cref{o-qdd-all} we prove that $\kcol\left(\ham q d\right)$ is $O(q^d/d)$ for all $q$. In \cref{omg-qdd} we prove that $\kcol\left(\ham q d\right)$ is $\Omega(q^d/d)$ in two different ways, completing the proof that $\kcol\left(\ham q d\right)$ is $\Theta(q^d/d)$ (\cref{main}), and we obtain as a corollary that $\kmarx\left(\ham q d\right)$ is $\Theta(q^d/d)$ as well. In \cref{tw-qdd} we summarize the proof of Chandran and Kavitha~\cite{CK06} that $\tw\left(\ham q d\right)$ is $\Theta\left(q^d\big/\sqrt d\right)$.
\subsection{Proof that \texorpdfstring{$\kcol(K_k) = k/4+O(1)$}{kappa(K\_k) = k/4 + O(1)}}\label{kcol-clique}

Rossman~\cite{Ros08} proved that $\kappa_{\Du}(K_k) \ge k/4$ (recall \cref{mc-unif}), so it suffices to prove the upper bound. By \cref{max-at-one} it suffices to prove that $\kappa_\Delta(K_k) \le k/4+O(1)$ for an arbitrary $\Delta = (1,\beta) \in \tcl(G)$. First we construct, by downwards induction, a sequence $U_1 \subseteq \dotsb \subseteq U_k = V(K_k)$ such that $U_i$ is an $i$-element subset of $V(K_k)$ and $\beta(K_k[U_i]) \ge \Bu(K_k[U_i])$ for all $i$. The set $U_k = V(K_k)$ satisfies this requirement because $\beta(K_k)$ and $\Bu(K_k)$ are both equal to $k$. Given $U_i$, let $\mathbf U_{i-1}$ be an $(i-1)$-element subset of $U_i$ chosen uniformly at random. Each pair of elements in $U_i$ is included in $\mathbf U_{i-1}$ with the same probability $p_i$ ($= 1-2/i$), so by linearity of expectation,
\begin{equation*}
\Ex[\beta(K_k[\mathbf U_{i-1}])] = \sum_{\mathclap{e \in E(K_k[U_i])}} \beta(e)p_i = p_i\beta(K_k[U_i]) \ge p_i\Bu(K_k[U_i]) = \Ex[\Bu(K_k[\mathbf U_{i-1}])].
\end{equation*}
Therefore there exists a fixed $U_{i-1}$ such that $\beta(K_k[U_{i-1}]) \ge \Bu(K_k[U_{i-1}])$.

Construct a union sequence $S$ for $K_k$ as follows: start by enumerating the edges, and then for $i$ from 1 to $k-1$, append $(K_k[U_i] \cup e_1, K_k[U_i] \cup e_1 \cup e_2, \dotsc, K_k[U_{i+1}])$, where $e_1, e_2, \dotsc$ are the edges between $U_i$ and $U_{i+1}-U_i$. Then,
\begin{equation*}
\max_{H \in S}\Delta(H) \le \max_i \Delta(K_k[U_i])+1 \le \max_i \Du(K_k[U_i]) + 1.
\end{equation*}
Finally, as observed in \cite{Ros08}, since $K_k$ is $(k-1)$-regular it follows from \cref{mc-eq} that
\begin{equation*}
	\Du(K_k[U_i]) = \frac{i(k-i)}{k-1} \le \frac{k^2}{4(k-1)} = \frac14\left(k + 1 + \frac{1}{k-1}\right) \le \frac{k+2}4 = k/4 + O(1).
\end{equation*}

\subsection{Proof that \texorpdfstring{$\kcol\left(\ham q d\right)$}{kappa(Kqd)} is \texorpdfstring{$O(q^d/d)$}{O(qd/d)} if \texorpdfstring{$q$}{q} is Even}\label{o-qdd-even}
First we reduce to the case $q=2$. The graph $\ham q d$ is a subgraph of $\mblowup {Q_d} {(q/2)^d}$ (recall \cref{emb-def}), as evidenced by the following argument. Let $\phi_L : [q] \rightarrow \bits$ and $\phi_R :[q] \rightarrow [q/2]$ such that $\phi_L \times \phi_R$ is a bijection from $[q]$ to $\bits \times [q/2]$, and let $\psi: [q/2]^d \rightarrow [(q/2)^d]$ be another arbitrary bijection. Then the following map is an injective homomorphism from $\ham q d$ to $\mblowup {Q_d} {(q/2)^d}$:
\begin{equation*}
(x_1,\dotsc,x_d) \mapsto ((\phi_L(x_1), \dotsc, \phi_L(x_d)), \psi(\phi_R(x_1), \dotsc, \phi_R(x_d))).
\end{equation*}
By \cref{kcol-mm,kcol-blowup}, if $\kcol(Q_d)$ is $O(2^d/d)$ then	
\begin{equation*}
\kcol\left(\ham q d\right) \le \kcol\left(\mblowup{Q_d}{(q/2)^d}\right) \le O\left(\left(\frac q 2\right)^d\kcol(Q_d)\right) \le O\left(\left(\frac q 2\right)^d \frac{2^d}d\right) = O(q^d/d).
\end{equation*}

Now we prove that $\kcol(Q_d)$ is $O(2^d/d)$, following some brief definitions and a high-level overview of the argument. Fix $d$. We identify each $u \in \bits^d$ with $\sum_{i=0}^{d-1} u_i 2^i$. For $0 \le a \le 2^d$ let $G(a) = Q_d[0, \dotsc, a-1]$. Recall that $\Du = (1,\Bu) = (1,2/d)$ is a threshold weighting on $Q_d$ (\cref{mc-unif}). Let $\mu = \max_{0\le a\le 2^d}\Du(G(a))$.
	
\begin{remark*}
	The intuition behind $\mu$ is as follows. The reader may note that $\kappa_{\Du}(Q_d) \le \mu+1$, by reasoning analogous to that in \cref{kcol-clique}. That is, for each vertex $u$ of $Q_d$ in increasing lexicographic order, add to an accumulator all edges $uv$ for which $v<u$.
	
	There is another union sequence captured by $\mu$ as well. If a subgraph $B \subseteq Q_d$ is isomorphic to $Q_k$ for some $k$, then since $Q_k$ is isomorphic to $G(2^k)$ (and $\Bu$ is uniform) it follows that $\Du(B) \le \mu$. Consider a depth-$d$ binary tree in which each node at depth $k$ is a subgraph of $Q_d$ isomorphic to $Q_{d-k}$ (in particular, the root is $Q_d$ and the leaves are vertices), and each interior node is the union of its two children along with some additional edges corresponding to a coordinate cut. This tree describes a union sequence $S$ for $Q_d$: recursively obtain the graphs $L$ and $R$ corresponding to the children of $Q_d$, and then take $L\cup R$ and add the missing edges. Note that $\max_{H\in S}\Du(H) = 2\max_{0\le k\le d}\Du(G(2^k)) \le 2\mu$.
	
	Analogous to \cref{kcol-clique}, the upper bound is obtained by comparing $\kappa_\Delta(Q_d)$ to $\mu$ for each $\Delta$, and bounding $\mu$. For this purpose we will consider the two union sequences mentioned above, as well as hybrids of them.
\end{remark*}

The proof is as follows:
\begin{align*}
\kcol(Q_d) &= \max_\beta\kappa_{(1,\beta)}(Q_d) &&\text{\cref{max-at-one}} \\
&\le 2\mu &&\text{\cref{tricky} (below) with $a=0$ and $k=d$} \\
&< 4/3\cdot2^d/d. &&\text{\cref{mu} (below)}
\end{align*}
	
For each threshold weighting $\Delta \in \tcl(Q_d)$, it will be convenient in the following to generalize $\kappa_\Delta$ to subgraphs $H\subseteq Q_d$ by $\kappa_\Delta(H) = \min_{S\in\seq(H)}\max_{F\in S}\Delta(F)$. (This is a nontrivial generalization of the definition of $\kappa_\Delta$, because if $\Delta(H) > 0$ then the restriction of $\Delta$ to subgraphs of $H$ is not a threshold weighting on $H$.) Also if $H$ is a single-vertex graph or the empty graph then let $\kappa_\Delta(H)=0$.
	
\begin{lem}\label{tricky}
	Let $0\le a\le 2^d$ and $0\le k\le d$ be such that $2^k$ divides $a$. Let $\Delta=(1,\beta)\in\tcl(Q_d)$ be such that $\beta(G(a))\ge\Bu(G(a))$ and $\beta(G(a+2^k)) \ge \Bu(G(a+2^k))$, and $\kappa_\Delta(G(a)) \le 2\mu$. Then $\kappa_\Delta(G(a+2^k)) \le 2\mu$.
\end{lem}
	
\begin{proof}		
	The proof is by induction on $k$. The inductive hypothesis will actually be (slightly) stronger in the following way: given a labeling of the vertices of $Q_d$ with the elements of $\bits^d$, the labels can be rearranged according to any of the $2^d d!$ isomorphisms of $Q_d$, and the inductive hypothesis is required to hold with respect to any such labeling. (The value of $\mu$ doesn't depend on the labeling used in its definition because of the symmetry of $\Bu$.)
		
	Let $B$ = $G(a+2^k)-G(a)$. Since $2^k$ divides $a$, it follows that $B$ is isomorphic to $Q_k$. In the inductive step we handle separately the cases where $\beta(B)\ge\Bu(B)$ and $\beta(B)<\Bu(B)$. The base case is a special case of the former because if $B$ is a single vertex then $\beta(B)$ and $\Bu(B)$ are both zero.
		
	\textbf{Case 1: $\beta(B)\ge\Bu(B)$.} If $k=0$ then $\kappa_\Delta(B) = 0 \le 2\mu$; we now obtain the same result in the case where $k>0$. For $0\le i<k$ and $b\in\bits$ let $B(i,b) = B[v\in V(B)\mid v_i=b]$. Choose $\mathbf i \in \{0, \dotsc, k-1\}$ and $\mathbf b \in \bits$ independently and uniformly at random. By symmetry, each $e \in E(B)$ is in $B(\mathbf i,\mathbf b)$ with the same probability $p$. (Specifically, $p=(k-1)/2k$: For any edge $uv$, there is a unique index $i$ in which $u$ and $v$ differ. If $\mathbf i = i$ then exactly one of $u$ and $v$ is in $B(\mathbf i, \mathbf b)$; otherwise $uv$ is in $B(\mathbf i, \mathbf b)$ with probability 1/2 depending on $\mathbf b$.) By linearity of expectation,
	\begin{equation*}
	\Ex[\beta(B(\mathbf i,\mathbf b))] = \sum_{\mathclap{e \in E(B)}}P(e \in B(\mathbf i, \mathbf b))\beta(e) = p\beta(B).
	\end{equation*}
	Similarly, $\Ex[\Bu(B(\mathbf i,\mathbf b))] = p\Bu(B)$. By our assumption that $\beta(B) \ge \Bu(B)$,
	\begin{equation*}
	\Ex[\beta(B(\mathbf i,\mathbf b))] = p\beta(B) \ge p\Bu(B) = \Ex[\Bu(B(\mathbf i,\mathbf b))].
	\end{equation*}
	Therefore there exist fixed $i$ and $b$ such that $\beta(B(i,b)) \ge \Bu(B(i,b))$.
		
	Now our claim that $\kappa_\Delta(B) \le 2\mu$ follows from two applications of the inductive hypothesis. Since we required the inductive hypothesis to hold for all labelings of $Q_d$, we can assume without loss of generality that $i = k-1$ and $b = 0$. Ignoring $G(a)$, an application of the inductive hypothesis with $a^\prime = 0$ and $k^\prime = k-1$ reveals that $\kappa_\Delta(B(i,b)) \le 2\mu$, and then a second application of the inductive hypothesis with $a^{\prime\prime} = 2^{k-1}$ and $k^{\prime\prime} = k-1$ reveals that $\kappa_\Delta(B) \le 2\mu$.
		
	Let $S$ be an optimal (with respect to $\Delta$) union sequence for $G(a)$, followed by an optimal union sequence for $B$, followed by $G(a) \cup B, G(a)\cup B\cup e_1, \dotsc, G(a)\cup B\cup \{e_j\}$, where the $\{e_j\}$ are the edges between $G(a)$ and $B$ in $Q_d$.\footnote{It is also necessary to add each edge $e_j$ individually to the union sequence, but clearly $\Delta(e_j) \le 2 \le 2(2-2/d) = 2\Du(G(1)) \le 2\mu$ if $d\ge2$, and if $d=1$ then the lemma holds trivially.} (If $G(a)$ or $B$ lacks edges then omit certain graphs from this sequence.) Then,
	\begin{equation*}
	\max_{H\in S}\Delta(H) \le \max(\kappa_\Delta(G(a)), \kappa_\Delta(B), \Delta(G(a))+\Delta(B)).
	\end{equation*}
		
	We proceed to bound each of these three terms by $2\mu$, completing the proof. We have assumed that $\kappa_\Delta(G(a)) \le 2\mu$, and proved that $\kappa_\Delta(B) \le 2\mu$. We have also assumed that $\beta(G(a)) \ge \Bu(G(a))$, and since $\Delta$ and $\Du$ both evaluate to 1 on all vertices, it follows that $\Delta(G(a)) \le \Du(G(a)) \le \mu$ (with the last step following from the definition of $\mu$). Similarly, since $B$ is isomorphic to $G(2^k)$ it follows that $\Delta(B) \le \Du(B) \le \mu$. Therefore $\Delta(G(a)) + \Delta(B) \le 2\mu$.
		
	\textbf{Case 2: $\beta(B)<\Bu(B)$.} For $i<k$ and $b \in \bits$ let $H(i,b) = Q_d[0,\dotsc,a-1,V(B(i,b))]$ (where $B(i,b)$ is defined as above). Choose $\mathbf i<k$ and $\mathbf b \in \bits$ independently and uniformly at random. Note that $\beta(G(a+2^k)) = \beta(G(a)) + \beta(G(a),B) + \beta(B)$.\footnote{Recall from \cref{tp} that $\beta(A,B) \coloneqq \sum_{u\in V(A), v \in V(B)} \beta(uv)$.} By reasoning similar to that in the previous case (and applying our various assumptions),
	\begin{align*}
	\Ex[\beta(H(\mathbf i,\mathbf b))] &= \beta(G(a)) + \frac12\beta(G(a),B) + \frac{k-1}{2k}\beta(B)\\
	&= \frac12\beta(G(a)) + \frac12\beta(G(a+2^k)) - \frac1{2k}\beta(B) \\
	&> \frac12\Bu(G(a)) + \frac12\Bu(G(a+2^k)) - \frac1{2k}\Bu(B) \\
	&= \Ex[\Bu(H(\mathbf i,\mathbf b))].
	\end{align*}
	Therefore $\beta(H(i,b)) > \Bu(H(i,b))$ for some fixed $i$ and $b$.
		
	Assume without loss of generality that $i=k-1$ and $b=0$; then $H(i,b) = G(a+2^{k-1})$. Applying the inductive hypothesis with $a^\prime = a$ and $k^\prime = k-1$ reveals that $\kappa_\Delta(G(a+2^{k-1})) \le 2\mu$, and then applying the inductive hypothesis with $a^{\prime\prime} = a+2^{k-1}$ and $k^{\prime\prime} = k-1$ reveals that $\kappa_\Delta(G(a+2^k)) \le 2\mu$.
\end{proof}
	
\begin{lem}\label{mu}
	$\mu < 2/3 \cdot 2^d/d$.
\end{lem}
\begin{proof}
	For any $0 \le a \le 2^d$, it follows from \cref{mc-eq} that $\Du(G(a)) = e(G(a), Q_d - G(a))/d$, so it suffices to prove that $e(G(a), Q_d - G(a)) < 2^{d+1}/3$ for all $a$. Let $G(a,b) = Q_d[a,\dotsc,b-1]$. Since
	\begin{equation*}
	e(G(0,a), G(a,2^d)) = e(G(0,2^d-a), G(2^d-a,2^d))
	\end{equation*}
	(as can be seen by applying the automorphism $(x_1,\dotsc,x_d) \mapsto (1-x_1,\dotsc,1-x_d)$ to $Q_d$), we can restrict our search to $a \in [0,2^{d-1}]$. In that case,
	\begin{align*}
	e(G(0,a), G(a,2^d)) &= e(G(0,a), G(a,2^{d-1})) + e(G(0,a), G(2^{d-1},2^d)) \\
	&= e(G(0,a), G(a,2^{d-1})) + v(G(0,a)) \\
	&= e(G(0,a), G(a,2^{d-1})) + a.
	\end{align*}
	By the same reasoning,
	\begin{equation*}
	e(G(0,2^{d-1}-a), G(2^{d-1}-a,2^d)) = e(G(0,2^{d-1}-a), G(2^{d-1}-a,2^{d-1})) + 2^{d-1}-a.
	\end{equation*}
	Since $e(G(0,a), G(a,2^{d-1})) = e(G(0,2^{d-1}-a), G(2^{d-1}-a,2^{d-1}))$ (consider a similar automorphism), it follows that if $a < 2^{d-2}$ then
	\begin{equation*}
	e(G(0,2^{d-1}-a),G(2^{d-1}-a,2^d)) - e(G(0,a),G(a,2^d)) = 2^{d-1} - 2a > 0.
	\end{equation*}
	Therefore we can restrict our search to $a \in [2^{d-2},2^{d-1}]$, in which case
	\begin{align*}
	e(G(0,a),G(a,2^d)) &= e(G(0,a), G(a,2^{d-1})) + a \\
	&= e(G(2^{d-2},a), G(a,2^{d-1})) + e(G(0,2^{d-2}), G(a,2^{d-1})) + a \\
	&= e(G(2^{d-2},a), G(a,2^{d-1})) + 2^{d-1}.
	\end{align*}
	By induction it follows that $\mu = 2^{d-1} + 2^{d-3} + 2^{d-5} + \dotsb + \text{(2 or 1)} < 2^{d+1}/3$.
\end{proof}

\begin{remark*}
	Harper~\cite{Har04} proved that out of all subgraphs of $Q_d$ with $a$ vertices, $G(a) = G(0,a)$ has the fewest outgoing edges~\cite{Fil15}.
\end{remark*}

\subsection{Proof that \texorpdfstring{$\kcol\left(\ham q d\right)$}{kappa(Kqd)} is \texorpdfstring{$\Omega(q^d/d)$}{Omega(qd/d)} and \texorpdfstring{$\kmarx\left(\ham q d\right)$}{emb(Kqd)} is \texorpdfstring{$\Theta(q^d/d)$}{Theta(qd/d)}}\label{omg-qdd}

Alon and Marx~\cite[Theorem 4.3]{AM11} proved that $\kmarx\left(\ham q d\right)$ is $\Omega(q^d/d)$, and it follows from \cref{secondary} that $\kmarx\left(\ham q d\right) \le O\left(\kcol\left(\ham q d\right)\right) \le O(q^d/d)$. Therefore $\kmarx\left(\ham q d\right)$ is $\Theta(q^d/d)$.

It is implicit in the above argument that $\kcol\left(\ham q d\right) \ge \Omega\left(\kmarx\left(\ham q d\right)\right) \ge \Omega(q^d/d)$; we now present an alternate proof that $\kcol\left(\ham q d\right)$ is $\Omega(q^d/d)$ based on edge expansion. Since $\ham q d$ is $d(q-1)$-regular, by \cref{kcol-cheeger} it suffices to prove that $1-\lambda_2\left(\ham q d\right)/d(q-1)$ is $\Omega(1/d)$. We use the following well-known fact, where graphs are identified with their adjacency matrices:
\begin{fact}\label{eig-prod}
	The eigenvalues of $G\cart H$ are $\lambda_i(G)+\lambda_j(H)$ for $i\in[v(G)], j\in[v(H)]$.
\end{fact}
\begin{proof}
	Observe that $G\cart H = G \otimes I + I \otimes H$, where the symbols $\otimes$ and $I$ denote the tensor product and the identity matrix respectively. Let $u_i$ (resp.\ $w_i$) be the $i$'th eigenvector of $G$ (resp.\ $H$); clearly $u_i\otimes w_j$ is an eigenvector of $G\cart H$ with eigenvalue $\lambda_i(G)+\lambda_j(H)$. Since a real symmetric matrix (in particular $G$ or $H$) has an orthogonal eigenbasis, it follows that the $u_i \otimes w_j$ are also orthogonal. Since $v(G\cart H) = v(G)v(H)$, there are no other eigenvalues of $G\cart H$.
\end{proof}

Since $\lambda_i(K_q)$ equals $q-1$ if $i=1$ and $-1$ otherwise, repeated application of \cref{eig-prod} reveals that $\lambda_2\left(\ham q d\right) = (q-1)(d-1)-1 = d(q-1)-q$, so $1-\lambda_2\left(\ham q d\right)/d(q-1) = q/(q-1)d$, as desired.

\begin{remark*}
	\cref{cheeger} is an equality in the case of hypercubes (see e.g.\ \cite{HLW06}): let $i\in[d]$ and define a cut by partitioning the vertices according to the values of their $i$'th coordinates. So for hypercubes, all slack in the application of \cref{kcol-cheeger} comes from \cref{kcol-expansion}.
\end{remark*}

\subsection{Proof that \texorpdfstring{$\tw\left(\ham q d\right)$}{tw(Kqd)} is \texorpdfstring{$\Theta\left(q^d/\sqrt d\right)$}{Theta(qd/root(d))}, Summarized}\label{tw-qdd}

(See \cite{CK06} for the full proof.) The proof that $\tw\left(\ham q d\right)$ is $O\left(q^d\big/\sqrt d\right)$ reduces to the case $q=2$ by reasoning analogous to that in the beginning of \cref{o-qdd-even}. For $k\in[d]$ let $U_k$ be the set of vertices of $Q_d$ with exactly $k$ or $k-1$ ones. The path $(U_1, \dotsc, U_d)$ is a tree decomposition of $Q_d$ with width approximately $2\binom d {d/2}$, and by Stirling's approximation this is $\Theta\left(2^d\big/\sqrt d\right)$.\footnote{Compared to the tree decomposition from \cite{CK06}, this one is a simpler variant whose width is larger by up to a constant factor.}

For a graph $G$ let $\phi(G)$ be the minimum over all $U \subseteq V(G), v(G)/4 \le |U| \le v(G)/2$ of the number of vertices in $V(G) - U$ with at least one neighbor in $U$. From a result of Robertson and Seymour~\cite{RS86} it follows that $\tw(G) \ge \phi(G) - 1$, and from a result of Harper~\cite{Har99} it follows that $\phi\left(\ham q d\right)$ is $\Omega\left(q^d\big/\sqrt d\right)$. (Also note the parallels between $\tw(G) \ge \phi(G) - 1$ and \cref{kcol-expansion}; interestingly, we've sign that both are tight to within a constant factor in the case of $\ham q d$.)

\section{\texorpdfstring{\aczero}{AC0} Upper Bound}\label{sec-5}
An \aczero circuit is a constant-depth circuit with unbounded-fanin $\textsf{AND}$ and $\textsf{OR}$ gates and $\textsf{NOT}$ gates. Fix a graph $G$ and threshold weighting $\Delta \in \tcl(G)$ for the remainder of this section. We prove the following, which is a more precise statement of \cref{asdf-aczero}:
\begin{thm}\label{circuit-main}
	There exists an \aczero circuit with $n^{\kappa_\Delta(G) + c}$ wires that solves $\subcol(\rg X \Delta n)$ with probability $1-n^{-\omega(1)}$, where $c>0$ is an absolute constant.
\end{thm}

Since in any circuit the number of gates is at most one plus the number of wires, the circuit from \cref{circuit-main} has size $n^{\kappa_\Delta(G)+O(1)} \le n^{\kcol(G)+O(1)}$. (In this discussion, all $\pm O(1)$ terms in an exponent are independent of $G$.) For comparison, it was proved in \cite{LRR17} (building on a line of previous work~\cite{Ros08, Ama10, Ros10, NW11}) that the average-case \aczero complexity of $\subcol(\rg X \Delta n)$ is between $n^{\kappa_\Delta(G)-o(1)}$ and $n^{2\kappa_\Delta(G)+O(1)}$. Another related result, regarding the uncolored $k$-clique problem, is that the average-case \aczero complexity of $\subuncol[K_k]\left(\ER n {n^{-2/(k-1)}}\right)$ is at most $n^{k/4+O(1)}$~\cite{Ama10, Ros14} ($= n^{\kcol(K_k)\pm O(1)}$ by \cref{kcol-clique}). See \cite{Ros18} for a survey of the average-case circuit complexity of subgraph isomorphism more generally.

One challenge to implementing the algorithm behind \cref{alg-main} in \aczero is that sorting cannot be done in (polynomial-size) \aczero~\cite{Has86}. The $n^{2\kappa_\Delta(G)+O(1)}$-size circuit from \cite{LRR17} computes $\colt {A\cup B} X$ by finding the relevant pairs in $\colt A X \times \colt B X$ by brute force with $\tilde O(|\colt A X| \cdot |\colt B X|)$ gates. Our circuit differs in that we represent $\colt H X$ as a depth-$v(H)$ tree, where the non-root vertices are assigned labels in $[n]$, and the (sequences of labels along the) paths from the root to the leaves correspond to the elements of $\colt H X$. This will allow us to compute $\colt {A \cup B} X$ with high probability given $\colt A X$ and $\colt B X$, on a circuit of size nearly linear in $|\colt A X| + |\colt B X|$. A key fact in our construction is that \aczero circuits can (with high probability) convert between representations of $\colt H X$ corresponding to different orderings of $V(H)$.

Our construction requires fairly precise estimates for how many children to assign each node. Luckily this number is highly concentrated around its mean if the input graph is $\rg X \Delta n$. This result will follow from the concentration inequality below, whose statement requires several definitions:

\begin{dfn}
	Let $X$ be in the support of $\rg X \Delta n$, and let $U\subseteq G$ be an arbitrary graph (which we think of as a ``universe"). Let $\colt U n$ be the set of all possible elements of $\colt U {\rg X \Delta n}$; note that this can be identified with $\prod_{v \in V(U)} [n^{\alpha(v)}]$. If $A\subseteq U$ and $\mathcal A \in \colt A n$ then let $\mathcal A$ \emph{extend to $U$ in $X$} if there exists a graph $\mathcal U \in \colt U X$ (called a \emph{$U$-extension of $\mathcal A$}) such that $\mathcal A \subseteq \mathcal U$. (In context, $X$ or $\mathbf X$ will be implicit.) Equivalently, $\mathcal A$ could be required to be in $\colt A X$ rather than $\colt A n$ in the latter definition.
	
	Let $\dst_U(A) = \min_{A\subseteq H\subseteq U}\Delta(H)$. Let $X$ be \emph{good} if for all graphs $U\subseteq G$ and $A\subseteq U$, and for all $\mathcal A \in \colt A n$ and vertices $v\in V(U)-V(A)$, there are $\tilde O\left(n^{\dst_U(A\cup v)-\dst_U(A)}\right)$ values of $i\in[n^{\alpha(v)}]$ such that $\mathcal A \cup v_i$ extends to $U$. (Recall our unconventional definition of $\tilde O(\cdot)$ from \cref{sec-2}, e.g.\ $\tilde O(1)$ denotes $\log^{O(1)}n$.) Finally, let an event occur \emph{with high probability (w.h.p.)}\ if it occurs with probability $1-n^{-\omega(1)}$.
\end{dfn}

We prove the following:
\begin{thm}\label{good-whp}
	The graph $\rg X \Delta n$ is good w.h.p.
\end{thm}

Observe that this is a substantially stronger concentration bound than the application of Markov's Inequality in the proof of \cref{alg-main}. In \cref{good-subsection} we prove \cref{good-whp}, and then in \cref{circuit-subsection} we use this result to prove \cref{circuit-main}. Both proofs use the following concentration inequality, which is proved by a Chernoff bound:

\begin{lem}\label{tail}
	If $\mathbf S = \mathbf S(n)$ is a sum of independent Bernoulli random variables then w.h.p.\ $\mathbf S \le \max(\Ex[\mathbf S],1) \cdot \tilde O(1)$.
\end{lem}
\begin{proof}
	Let $\mathbf S = \sum_i \mathbf B_i$ be a decomposition of $\mathbf S$ as a sum of independent Bernoulli random variables. Let $p_i = \Ex[\mathbf B_i]$ and $\mu = \Ex[\mathbf S] = \sum_i p_i$. Then for $r,t \ge 0$,
	\begin{align*}
	P(\mathbf S\ge r) &= P(\exp(t\mathbf S) \ge \exp(tr)) \le \exp(-tr)\prod_i \Ex[\exp(t\mathbf B_i)] = \exp(-tr)\prod_i(1-p_i+p_i e^t) \\
	&\le \exp(-tr)\prod_i \exp(p_ie^t) = \exp(-tr + \mu e^t).
	\end{align*}
	Letting $t = \log(r/\mu)$ gives $P(\mathbf S \ge r) \le (e\mu/r)^r$ assuming $t \ge 0$, and then (for example) letting $r = \max(\mu,1) \log^2 n$ gives, for sufficiently large $n$,
	\begin{equation*}
	P(\mathbf S \ge r) \le (e/\log^2 n)^{\log^2 n} \le (1/e)^{\log^2 n} = n^{-\log n}. \qedhere
	\end{equation*}
\end{proof}
\subsection{Proof of \texorpdfstring{\cref{good-whp}}{Theorem 6.3}}\label{good-subsection}

First we derive some algebraic properties of the threshold weighting $\Delta$.
\begin{lem}\label{Delta-sum}
	If $A,B\subseteq G$ then $\Delta(A) + \Delta(B) = \Delta(A\cap B) + \Delta(A\cup B)$.
\end{lem}
\begin{proof}
	Each vertex or edge in one (resp.\ two) of $A$ and $B$ is also in one (resp.\ two) of $A\cap B$ and $A\cup B$.
\end{proof}

\begin{dfn}
	For $A\subseteq U\subseteq G$ let $\Gamma_U(A) = \bigcap\{H \in [A,U] \mid \Delta(H) = \dst_U(A)\}$, and let $A$ be a \emph{$U$-base} if $\Delta(A) = \dst_U(A)$.
\end{dfn}

Throughout this subsection, $U$ will be an arbitrary subgraph of $G$ unless additional structure is imposed on it, and missing subscripts on $\dst$ and $\Gamma$ default to $U$.

\begin{lem}\label{Gamma-lemma-1}
	If $A\subseteq U$ then $\Delta(\Gamma(A)) = \dst(A)$ and $A\subseteq\Gamma(A)$.
\end{lem}
\begin{proof}
	It suffices to show that the set $S = \{H \in [A,U] \mid \Delta(H) = \dst(A)\}$ is closed under intersection. Let $B, C \in S$. By the definition of $S$, \cref{Delta-sum}, and the fact that $A\subseteq B\cup C$,
	\begin{equation*}
	2\dst(A) = \Delta(B) + \Delta(C) = \Delta(B\cap C) + \Delta(B\cup C) \ge \Delta(B\cap C) + \dst(A),
	\end{equation*}
	so $\Delta(B\cap C) \le \dst(A)$. On the other hand, $\Delta(B\cap C) \ge \dst(A)$ because $A\subseteq B\cap C$. Therefore $\Delta(B\cap C) = \dst(A)$, so $B\cap C \in S$.
\end{proof}

\begin{lem}\label{gb}
	If $A\subseteq \Gamma(A) \subseteq U^\prime \subseteq U$ then $\Gamma(A)$ is a $U^\prime$-base.
\end{lem}
\begin{proof}
	Since the interval $[A,U]$ includes the interval $[\Gamma(A), U^\prime]$, it follows from \cref{Gamma-lemma-1} that $\Delta(\Gamma(A)) = \dst_U(A) \le \dst_{U^\prime}(\Gamma(A)) \le \Delta(\Gamma(A))$. Therefore $\Delta(\Gamma(A)) = \dst_{U^\prime}(\Gamma(A))$.
\end{proof}

\begin{lem}\label{Gamma-subset}
	If $A\subseteq B\subseteq U$ then $\Gamma(A) \subseteq \Gamma(B)$.
\end{lem}
\begin{proof}
	Since $B \subseteq \Gamma(B) \subseteq \Gamma(B) \cup \Gamma(A)$ it follows that $\dst(B) \le \Delta(\Gamma(B) \cup \Gamma(A))$, so by \cref{Gamma-lemma-1,Delta-sum},
	\begin{align*}
	\Delta(\Gamma(A)\cap\Gamma(B)) + \dst(B)
	&\le \Delta(\Gamma(A)\cap\Gamma(B)) + \Delta(\Gamma(A)\cup\Gamma(B)) \\
	&= \Delta(\Gamma(A)) + \Delta(\Gamma(B)) \\
	&= \dst(A) + \dst(B).
	\end{align*}
	Therefore $\dst(A) \ge \Delta(\Gamma(A)\cap\Gamma(B))$. On the other hand, since $A\subseteq \Gamma(A)$ and $A\subseteq B\subseteq \Gamma(B)$ it follows that $A\subseteq \Gamma(A)\cap\Gamma(B)$, so $\dst(A) \le \Delta(\Gamma(A)\cap\Gamma(B))$. Therefore $\dst(A) = \Delta(\Gamma(A)\cap\Gamma(B))$, so it follows from the definition of $\Gamma(A)$ that $\Gamma(A) \subseteq \Gamma(A)\cap\Gamma(B) \subseteq \Gamma(B)$.
\end{proof}

We now analyze the concentration of $\rg X \Delta n$, making liberal use of the fact that if $n^{O(1)}$ events occur with uniformly high probability then their conjunction also occurs w.h.p.\ by a union bound. For the rest of this subsection, ``extensions" are with respect to an implicit $\mathbf X \equiv \rg X \Delta n$.

\begin{lem}\label{atom-extensions}
	If $A\subseteq U$ and $\Gamma_U(A) = U$ (i.e. $\Delta(H)>\Delta(U)$ for all $H\in[A,U)$) then the number of $U$-extensions of any $\mathcal A\in \colt A n$ is $\tilde O(1)$ w.h.p.
\end{lem}
(The above conditions are equivalent because, by the definition of $\Gamma(A)$, we have $\Gamma(A) = U$ if and only if $U$ is the unique $H\in[A,U]$ that minimizes $\Delta(H)$.)
\begin{proof}
	The result is trivial for $A=U$; assume it's true for all $B\in(A,U]$ and that $A\neq U$. (Since $\Delta(H) > \Delta(U)$  for all $H \in [A,U)$, for any $B \in (A,U]$ it is the case that $\Delta(H) > \Delta(U)$ for all $H \in [B,U)$.) Assume without loss of generality that $A=U[V(A)]$, since all $U$-extensions of $\mathcal A$ are also $U$-extensions of $\mathcal A$'s unique possible $U[V(A)]$-extension. Also condition on $\mathcal A \subseteq \mathbf X$, since otherwise $\mathcal A$ trivially has zero $U$-extensions.
	
	There are $n^{\alpha(U)-\alpha(A)}$ possible $U$-extensions of $\mathcal A$, so there are at most $n^{(\alpha(U)-\alpha(A))\log n}$ sets of $\log n$ possible $U$-extensions of $\mathcal A$ whose projections onto $\colt {U-A} n$ are pairwise vertex-disjoint. (This is true even if we omit the condition about vertex-disjointness.) For each of these sets, all of its elements are subgraphs of $\mathbf X$ with probability $n^{(-\beta(U)+\beta(A))\log n}$, so this occurs for at least one such set with probability at most $n^{(\Delta(U)-\Delta(A))\log n}$ (by a union bound). By assumption, $\Delta(U)-\Delta(A)<0$, so w.h.p.\ any set of $U$-extensions of $\mathcal A$ whose projections onto $\colt {U-A} n$ are pairwise vertex-disjoint has $\tilde O(1)$ elements.
	
	Let $S$ be one such set, such that $S$ is maximal. It follows that every $U$-extension of $\mathcal A$ agrees with some element of $S$ on some vertex in $V(U)-V(A)$. Therefore $\mathcal A$ has at most $\sum_{\mathcal U\in S}\sum_{H\in(A,U]}E(\mathcal U,H)$ $U$-extensions, where $E(\mathcal U,H)$ is the number of $U$-extensions of $\mathcal A$ that agree with $\mathcal U$ on precisely $H$. By the inductive hypothesis, $E(\mathcal U,H)$ is $\tilde O(1)$ w.h.p.\ for all $\mathcal U$ and $H$ (independent of $S$), so $\mathcal A$ has $\tilde O(1)$ $U$-extensions w.h.p.\ by a union bound.
\end{proof}

\begin{lem}\label{base-extensions}
	If $A$ is a $U$-base then any $\mathcal A\in \colt A n$ has $\tilde O(n^{\Delta(U)-\Delta(A)})$ $U$-extensions w.h.p.
\end{lem}
\begin{proof}
	Again, assume that $A$ is an induced subgraph of $U$ and condition on $\mathcal A \subseteq \mathbf X$. Also assume without loss of generality that $\beta$ is strictly positive on $E(U)$. The proof is by induction on $v(U)-v(A)$, for all $U \subseteq G$. The base case $A = U$ is trivial. Fix an arbitrary vertex $v \in V(U)-V(A)$. First we consider the case where $\Gamma(A\cup v) \neq U$. The number of $U$-extensions of $\mathcal A$ equals the sum over all $\gamma \in \{\text{$\Gamma(A\cup v)$-extensions of $\mathcal A$}\}$ of the number of $U$-extensions of $\gamma$. Clearly $A$ is a $\Gamma(A\cup v)$-base, and \cref{gb} implies that $\Gamma(A\cup v)$ is a $U$-base. It follows from our assumptions that $v(A) < v(\Gamma(A\cup v)) < v(U)$, so we can apply the inductive hypothesis twice: w.h.p.\ $\mathcal A$ has $\tilde O(n^{\Delta(\Gamma(A\cup v)) - \Delta(A)})$ extensions to $\Gamma(A\cup v)$, each of which has $\tilde O(n^{\Delta(U) - \Delta(\Gamma(A\cup v))})$ extensions to $U$, and the result follows.
	
	Now assume that $\Gamma(A\cup v) = U$. \cref{atom-extensions} implies that $\mathcal A \cup v_i$ has $\tilde O(1)$ $U$-extensions w.h.p.\ for any $i$, so it suffices to show that w.h.p.\ there are $\tilde O(n^{\Delta(U)-\Delta(A)})$ values of $i$ such that $\mathcal A \cup v_i$ extends to $U$. Let $\mathbf W = \mathbf X[u_i \mid u \in V(U)-v, i \in [n^{\alpha(u)}]]$, and if $\mathcal A$ has $\tilde O\left(n^{\Delta(U-v)-\Delta(A)}\right)$ extensions to $U-v$ when $\mathbf W=W$ then let $W$ be ``okay". Since $A$ is a $(U-v)$-base, $\mathbf W$ is okay w.h.p.\ by the inductive hypothesis. Let $\mathbf Z_i = \Ind{\text{$\mathcal A \cup v_i$ extends to $U$}}$, and let $E$ be the event that $\sum_i\mathbf Z_i > \tilde O(n^{\Delta(U)-\Delta(A)})$. Then,
	\begin{align*}
	P(E) &= P(E\mid\text{$\mathbf W$ is okay})P(\text{$\mathbf W$ is okay}) + P(E\mid\text{$\mathbf W$ isn't okay})P(\text{$\mathbf W$ isn't okay})\\
	&\le P(E\mid\text{$\mathbf W$ is okay}) + P(\text{$\mathbf W$ isn't okay})\\
	&\le \max_{\text{okay $W$}}P(E\mid \mathbf W = W) + n^{-\omega(1)},
	\end{align*}
	so it suffices to prove that $P(E\mid\mathbf W=W) \le n^{-\omega(1)}$ for all okay $W$.
	
	The $\mathbf Z_i$ are independent Bernoulli random variables (given $W$). By a union bound, $\Ex[\mathbf Z_i]$ is at most the number of $(U-v)$-extensions of $\mathcal A$ times the probability that the requisite edges between any one of them and $v_i$ are in $\mathbf X$, i.e.\
	\begin{equation*}
	\Ex[\mathbf Z_i] \le \tilde O\left(n^{\Delta(U-v)-\Delta(A)}\right)n^{\beta(U-v)-\beta(U)} = \tilde O\left(n^{\Delta(U)-\Delta(A)-\alpha(v)}\right).
	\end{equation*}
	Since $A$ is a $U$-base, $\Delta(U) - \Delta(A) \ge 0$, so it follows from \cref{tail} that $\sum_{i=1}^{n^{\alpha(v)}} \mathbf Z_i$ is $\tilde O(n^{\Delta(U)-\Delta(A)})$ w.h.p.
\end{proof}

\begin{remark*}
	It follows from \cref{janson-app} that \cref{base-extensions} is essentially tight.
\end{remark*}

Now we prove that $\rg X \Delta n$ is good w.h.p.:
\begin{proof}[Proof of \cref{good-whp}]
	Let $A \subseteq U$, $\mathcal A \in \colt A n$ and $v \in V(U) - V(A)$. By a union bound it suffices to prove that w.h.p.\ there are $\tilde O(n^{\dst(A\cup v) - \dst(A)})$ values of $i$ such that $\mathcal A \cup v_i$ extends to $U$. The number of such $i$ is at most the number of $i$ such that $\mathcal A\cup v_i$ extends to $\Gamma(A\cup v)$, which is at most the number of $\Gamma(A\cup v)$-extensions of $\mathcal A$. Since $\Gamma(A) \subseteq \Gamma(A\cup v)$ (\cref{Gamma-subset}), this equals the sum over all $\gamma \in \{\text{$\Gamma(A)$-extensions of $\mathcal A$}\}$ of the number $\mathbf E_\gamma$ of $\Gamma(A\cup v)$-extensions of $\gamma$.
	
	It follows from \cref{atom-extensions} that $\mathcal A$ has $\tilde O(1)$ extensions to $\Gamma(A)$ w.h.p.\ (To see this, note that if $A \subseteq H \subset \Gamma(A)$ then $\Delta(H) \ge \dst(A) = \Delta(\Gamma(A))$ (\cref{Gamma-lemma-1}), and if $\Delta(H) = \dst(A)$ then it follows from the definition of $\Gamma(A)$ that $\Gamma(A) \subseteq H$, a contradiction.) Since $\Gamma(A)$ is a $\Gamma(A\cup v)$-base (\cref{gb}), it follows from \cref{base-extensions} that any $\mathbf E_\gamma$ is $\tilde O(n^{\Delta(\Gamma(A\cup v)) - \Delta(\Gamma(A))})$ w.h.p.\ ($= \tilde O(n^{\dst(A\cup v) - \dst(A)})$ by \cref{Gamma-lemma-1}).
\end{proof}

\subsection{The Circuit}\label{circuit-subsection}

If $D$ is a data structure then let $|D|$ denote the number of bits used to represent it according to whatever schema we describe. If $A$ is a bit array and $b$ is a bit then let $(A\vee b)_i = A_i\vee b$ and $(A\wedge b)_i = A_i \wedge b$ for all $i\in[|A|]$. When there is a null element we represent it by the all-zeros string.

We now prove \cref{circuit-main}, i.e.\ that there exists an \aczero circuit with $\tilde O(n^{\kappa_\Delta(G)+3})$ wires that solves $\subcol(\rg X \Delta n)$ w.h.p. Since $\rg X \Delta n$ is good w.h.p.\ (\cref{good-whp}) it suffices to prove the existence of a small \aczero circuit $\mathsf C$ such that $P_{X\sim\rg X \Delta n}(\mathsf C(X)=\subcol(X)\mid \text{$X$ is good}) = 1-n^{-\omega(1)}$. By Yao's Principle~\cite{Yao77} it suffices to prove the existence of a small, random \aczero circuit $\mathbf C$ such that $P(\mathbf C(X)=\subcol(X))=1-n^{-\omega(1)}$ for all fixed good $X$. More precisely,
\begin{align*}
	\max_{\mathsf C : P(\mathbf C = \mathsf C)>0}P_{X\sim\rg X \Delta n} \left(\mathsf C(X) = \subcol(X)\right) \ge \\
	\max_{\mathsf C : P(\mathbf C = \mathsf C)>0}P_{X\sim\rg X \Delta n} \left(\mathsf C(X) = \subcol(X) \mid \text{$X$ is good}\right) P(\text{$\rg X \Delta n$ is good})\ge \\
	P_{X\sim\rg X \Delta n} \left(\mathbf C(X) = \subcol(X) \mid \text{$X$ is good}\right) \cdot \left(1-n^{-\omega(1)}\right)\ge \\
	\left(1-n^{-\omega(1)}\right)\min_{\text{good $X$}}P(\mathbf C(X)=\subcol(X)).
\end{align*}
	
The following result is essentially implicit in \cite{LRR17} (as is the argument above) and helps keep the random circuit small:
	
\begin{lem}[Random Hashing]\label{hash}
	Let $S$ be a set containing a null element, and assume all elements of $S$ are represented using the same number of bits. Let $l = l(n) \le n^{O(1)}$ and $m = m(n)$ be functions of $n$. Then there exists a random \aczero circuit $\mathbf C : S^l \rightarrow S^{\tilde O(m)}$ such that if $A$ is an array of $l$ values in $S$, of which all but at most $m$ are null, then $\mathbf C$ has at most $|A|n^{o(1)}$ gates and $|A|\tilde O(l/m)$ wires, and w.h.p.\ the multiset of non-null elements of $\mathbf C(A)$ is the same as that of $A$.
\end{lem}
We remark that \cref{hash} will only be called with $l \le \tilde O(n)$.
\begin{proof}
	The result is trivial if $l \le m$ (simply return $A$) so assume otherwise. Let $\mathbf h:[l]\rightarrow[m]$ be a uniform random function. Let $\mathbf B$ be an $m\times \tilde O(1)$ array of values in $S$, where $\mathbf B[p,q]$ is the $q$'th non-null element of $\mathbf A^{(p)} \coloneqq A[\mathbf h^{-1}(p)]$ if this set has at least $q$ elements, and $\mathbf B[p,q]$ is null otherwise. Each of the at most $m$ non-null elements of $A$ is independently in $\mathbf A^{(p)}$ with probability $1/m$, so for any particular $p$, the sub-array $\mathbf B[p,:]$ is large enough to store the non-null elements of $\mathbf A^{(p)}$ w.h.p.\ (\cref{tail}). It follows from a union bound that $\mathbf B$ has the same non-null elements as $A$ w.h.p. Also assume that $|\mathbf h^{-1}(p)|$ is $\tilde O(l/m)$ for all $p$; this occurs w.h.p.\ by \cref{tail}. Under these conditions it suffices to compute $\mathbf B$, and this can be done as follows.
		
	For $x\in\bits^N$ let $T_k^N(x) = \Ind{\text{$x$ has at least $k$ ones}}$. Then $\mathbf J^{(p)}[i] \coloneqq \Ind{\mathbf A^{(p)}[i] \neq \nul}$ can be computed by applying a single $\textsf{OR}$ gate to all elements of $\mathbf A^{(p)}[i]$, and
	\begin{equation*}
	\mathbf B[p,q] = \bigvee_{\mathclap{i\in[|\mathbf h^{-1}(p)|]}} \left(T_q^i\left(\mathbf J^{(p)}[1:i]\right) \wedge \neg T_q^{i-1}\left(\mathbf J^{(p)}[1:i-1]\right) \wedge \mathbf A^{(p)}[i]\right).
	\end{equation*}
		
	\begin{fact}[{\cite[Theorem 6]{Has+94}}]\label{threshold}
		If $k=\floor{\log^\gamma N}$ for constant $\gamma$, then $T_k^N$ can be computed for $m=\floor\gamma+1$ by monotone unbounded fan-in circuits of depth $m+2$ with $2^{O(\log^{\gamma/m}N\log\log N)}$ gates, where $\gamma/m<1$, and $O(N\log^{2\gamma+2}N)$ wires.
	\end{fact}
	
	Let $N = \tilde O(l/m) = n^{O(1)}$, and let $\gamma$ be a constant such that the dimensions of $\mathbf B$ are at most $m \times k$ where $k = \floor{\log^\gamma N}$. Let $\mathsf T$ be the $N^{o(1)}$-size (hence $n^{o(1)}$-size) circuit from \cref{threshold} that computes $T_k^N$. Observe that $T_q^i(x) = \mathsf T(x,y)$ where $y \in \bits^{N-i}$ is an arbitrary fixed string with exactly $k-q$ ones that can be hard-coded in. Therefore $\mathbf B[p,q]$ can be computed by an \aczero circuit of size $\sum_{i\in[|\mathbf h^{-1}(p)|]}\left(n^{o(1)}+|\mathbf A^{(p)}[i]|\right) \le \left|\mathbf A^{(p)}\right|n^{o(1)}$. Summing over $p$ and $q$, the total number of gates is $|A|n^{o(1)}$. To count wires instead of gates, replace $n^{o(1)}$ with $\tilde O(N) = \tilde O(l/m)$.
\end{proof}

Given $H\subseteq G$ and an ordering $\pi = (\pi^1,\dotsc,\pi^{v(H)})$ of $V(H)$, we can represent $\colt H X$ as a tree in the following way. Start with a rooted, depth-$v(H)$ tree (meaning the root has depth 0 and the leaves have depth $v(H)$) in which each interior node has $n$ unordered children labeled $1, \dotsc, n$. Then take the induced subtree of this tree on the union of all root-to-leaf paths $(\mathrm{root}, l_1, \ldots, l_{v(H)})$ such that\footnote{Recall that $(\pi^j)_{l_j}$ is a $\pi^j$-colored vertex in $X$.} $\pi^1_{l_1}, \ldots, \pi^{v(H)}_{l_{v(H)}}$ are the vertices of an $H$-colored subgraph of $X$.
	
With respect to an implicit $H$ and $\pi$, let $\delta_i = \dst_H(\pi^1 \cup \dotsb \cup \pi^i)$ for $0 \le i \le v(H)$, and let $\phi_i = \delta_{i+1}-\delta_i$ for $0\le i<v(H)$.
	
\begin{lem}\label{phi-bound}
	$0 \le \phi_i \le 1$ for all $i$.
\end{lem}
\begin{proof}
	Clearly $\delta_i \le \delta_{i+1}$. Let $A \subseteq H$ such that $\pi^1, \dotsc, \pi^i \in V(A)$ and $\Delta(A) = \delta_i$. Then $\delta_{i+1} \le \Delta(A \cup \pi^{i+1}) \le \Delta(A) + \alpha(\pi^{i+1}) \le \delta_i + 1$.
\end{proof}

Let $T=T(H,\pi)$ be a depth-$v(H)$ tree in which each node at depth $i<v(H)$ has $n^{\phi_i}\log^{c_i}n$ children, where $c_i$ is a sufficiently large constant. Each non-root node $N$ has a label $\mathcal L(N)\in\{\nul\}\cup[n]$, and the root is labeled ``root". It is required that if we ignore the null nodes of $T$, then $T$ is isomorphic to the tree representation of $\colt H X$ described above. 
	
If the underlying tree structure of $T$ (that is, everything except the labels) is implicit, then we can represent $T$ by an array of values in $\{\nul\} \cup [n]$, indexed by the nodes of $T$. Each of these values can be associated with a bit string in a natural way. We will consider circuits that compute $T$ according to this representation.
	
Let $S$ be an \emph{immediate subtree} of $T$ (resp.\ of a node $N$), denoted $S\in T$ (resp.\ $S\in N$), if $S$'s root is a child of $T$'s root (resp.\ of $N$). Any subtree is considered to have the same label as its root.
	
\begin{lem}\label{tree-size}
	$|T|$ is $\tilde O(n^{\Delta(H)})$.
\end{lem}
\begin{proof}
	$\delta_0 = \Delta(\emptyset) = 0$ and $\delta_{v(H)}=\dst_H(V(H))=\Delta(H)$. It takes $\tilde O(1)$ bits to store an element of $[n]^{V(H)}$, and each $\phi_i$ is nonnegative (\cref{phi-bound}), so
	\begin{equation*}
		|T| = \tilde O\left(\prod_{i=0}^{v(H)-1} n^{\phi_i}\right) = \tilde O\left(n^{\sum_{i=0}^{v(H)-1}\phi_i}\right) = \tilde O\left(n^{\delta_{v(H)}-\delta_0}\right) = \tilde O\left(n^{\Delta(H)}\right). \qedhere
	\end{equation*}
\end{proof}
	
\begin{lem}\label{sort}
	For all $H\subseteq G$ and orderings $\pi,\pi^\prime$ of $V(H)$, there exists a random \aczero circuit, independent of $X$, with $\tilde O(n^{\Delta(H)+2})$ wires, that computes $T(H,\pi^\prime)$ from $T(H,\pi)$ w.h.p.
\end{lem}
	
\begin{proof}
	Assume that $\pi$ and $\pi^\prime$ differ only in positions $d$ and $d+1$. (The general case can be reduced to at most $\binom{v(H)}2$ copies of this circuit in succession.) Define $\delta_i^\prime$ and $\phi_i^\prime$ analogously to $\delta_i$ and $\phi_i$, but with respect to $\pi^\prime$ rather than $\pi$. Clearly $\delta_i = \delta_i^\prime$ for $i\neq d$, so $\phi_i = \phi_i^\prime$ for $i\notin\{d-1,d\}$.
		
	For each depth-$(d-1)$ node $N$ of $T(H,\pi)$, in parallel, do the following. For $\sigma \in N, j \in [n]$ let $\tau^\prime_{\sigma j} = \bigvee_{\tau \in \sigma} \left((\mathcal L(\tau) = j) \wedge \tau^{(\mathcal L(\sigma))}\right)$, where $\tau^{(\mathcal L(\sigma))}$ is formed from $\tau$ by replacing its (root's) label with $\mathcal L(\sigma)$. Let $\sigma^\prime_j$ be the tree whose immediate subtrees are $\tau^\prime_{\sigma j}$ for $\sigma \in N$, and whose label is $\left(\bigvee_{\sigma\in N}\bigvee_{\tau\in\sigma} (\mathcal L(\tau)=j)\right) \wedge \overline j$ where $\overline j$ is the bit-string representation of $j$. Hash the number of immediate subtrees of $\sigma_j^\prime$ down to $\tilde O(n^{\phi_d^\prime})$ for each $j$ in parallel, and hash the number of $\sigma_j^\prime$ down to $\tilde O(n^{\phi_{d-1}^\prime})$. (The hashing uses \cref{hash} and succeeds w.h.p.\ because $X$ is good; also note that $\phi_d+\phi_{d-1} = \delta_{d+1} - \delta_{d-1} = \phi_d^\prime+\phi_{d-1}^\prime$.) Finally, the new children of $N$ are the remaining $\sigma^\prime_j$.
		
	Computing $\tau^\prime_{\sigma j}$ takes $\tilde O(\sum_{\tau\in\sigma}|\tau|) = \tilde O(|\sigma|)$ wires, so computing $\sigma^\prime_j$ takes $\tilde O(\sum_{\sigma \in N}|\sigma|) = \tilde O(|N|)$ wires, and doing this for all $N$ and $j$ takes $\tilde O(n|T|) = \tilde O(n^{\Delta(H)+1})$ wires (\cref{tree-size}). The hashing increases the number of wires by a factor of $\tilde O(n)$.
\end{proof}

For $uv \in E(G)$ we can construct $T(uv)$ as follows. Suppose we're given the adjacency matrix $A \in \bits^{n^{\alpha(u)} \times n^{\alpha(v)}}$ such that $A_{ij} = \Ind{u_iv_j\in E(X)}$. Let $\tau^\prime_{ij} = A_{ij} \wedge \overline i$. Let $\sigma^\prime_j$ be the tree with children $\tau^\prime_{ij}$ for $i \in [n^{\alpha(u)}]$, and label $\left(\bigvee_i A_{ij}\right) \wedge \overline j$. This setup is equivalent to the situation immediately before the hashing in the proof of \cref{sort}, and the rest of the construction is the same. This takes $\tilde O(n^3)$ wires, including the hashing.
	
\begin{lem}\label{merge}
	For all $H, H^\prime \subseteq G$ and orderings $\pi$ and $\pi^\prime$ of $V(H)$ and $V(H^\prime)$ respectively, there exists a random \aczero circuit, independent of $X$, with $\tilde O(n^{\max(\Delta(H), \Delta(H^\prime)) + 2})$ wires, that computes $T(H\cup H^\prime,\hat\pi)$ from $T(H,\pi)$ and $T(H^\prime,\pi^\prime)$ w.h.p.\ for some $\hat \pi$.
\end{lem}

\begin{proof}
	Let $T = T(H,\pi)$ and $T^\prime = T(H^\prime,\pi^\prime)$. By \cref{sort} we can assume without loss of generality that $\{\pi^1, \dotsc, \pi^{v(H\cap H^\prime)}\} = \{\pi^{\prime1}, \dotsc, \pi^{\prime v(H\cap H^\prime)}\} = V(H\cap H^\prime) = V(H) \cap V(H^\prime)$, and that $\pi^k = \pi^{\prime k} = \hat\pi^k$ for $k \in [v(H\cap H^\prime)]$. Define $\phi^\prime$ and $\hat\phi$ with respect to $(H^\prime, \pi^\prime)$ and $(H\cup H^\prime, \hat \pi)$ respectively.
		
	Let $\psi_i = \min(\phi_i,\phi_i^\prime)$. For $0 \le d \le v(H\cap H^\prime)$ let $S_d$ be a depth-$d$ tree in which each node at depth $i<d$ (including $i=0$) has $\tilde O(n^{\psi_i})$ children. Again, each non-root node $N$ of $S_d$ has a label $\mathcal L(N)\in\{\nul\}\cup[n]$, and the root is labeled ``root". It is required that if we ignore null nodes, then $S_d$ is isomorphic to the intersection of the depth-$d$ truncations of $T$ and $T^\prime$. Furthermore, each leaf $\ell$ of $S_d$ is associated with the pair $(\sigma,\sigma^\prime)$ of subtrees of $T$ and $T^\prime$ such that the $\mathrm{root}(S)$-to-$\ell$ path in $S$, the $\mathrm{root}(T)$-to-$\mathrm{root}(\sigma)$ path in $T$, and the $\mathrm{root}(T^\prime)$-to-$\mathrm{root}(\sigma^\prime)$ path in $T^\prime$ are all the same sequence of labels.
		
	The tree $S_0$ is the single node $(T,T^\prime)$, and we can compute $S_{d+1}$ from $S_d$ by doing the following for each leaf $(\sigma,\sigma^\prime)$ of $S_d$ in parallel. Assume without loss of generality that $\psi_d=\phi_d$. (If $\psi_d = \phi_d^\prime$, reverse the roles of $\sigma$ and $\sigma^\prime$ in the following construction.) For $\tau\in\sigma$ let $\rho_\tau$ be the immediate subtree of $\sigma^\prime$ with the same label as $\tau$ (if this exists), i.e.\ $\rho_\tau = \bigvee_{\tau^\prime\in\sigma^\prime} ((\mathcal L(\tau)=\mathcal L(\tau^\prime)) \wedge \tau^\prime)$. Replace $(\sigma,\sigma^\prime)$ with a new node with children $(\rho_\tau\neq\nul) \wedge (\tau,\rho_\tau)$ for all $\tau\in\sigma$. Assign the node replacing $(\sigma,\sigma^\prime)$ the same label as $(\sigma,\sigma^\prime)$, and assign $(\tau,\rho_\tau)$ the same label as $\tau$ and $\rho_\tau$.
		
	Computing $\rho_\tau$ takes $\tilde O\left(\sum_{\tau^\prime\in\sigma^\prime} |\tau^\prime|\right) = \tilde O(|\sigma^\prime|)$ wires, and there are at most $n$ values of $\tau$ (\cref{phi-bound}), so computing $\rho$ takes $\tilde O(n|\sigma^\prime|)$ wires. Given $\rho$, computing the leaves of the replacement for $(\sigma,\sigma^\prime)$ takes $O\left(\sum_{\tau\in\sigma}(|\tau|+|\rho_\tau|)\right) = O(|\sigma|+|\sigma^\prime|)$ wires. Since the roles of $\sigma$ and $\sigma^\prime$ might be reversed above, all of this takes at most $\tilde O(n|\sigma| + n|\sigma^\prime|)$ wires. Since $S_d$ has $\tilde O\left(n^{\sum_{i<d}\psi_i}\right)$ leaves, the number of wires is at most
	\begin{equation*}
	\tilde O\left(n^{\sum_{i<d}\psi_i}(n|\sigma|+n|\sigma^\prime|)\right) \le \tilde O\left(n^{1+\sum_{i<d}\phi_i}|\sigma| + n^{1+\sum_{i<d}\phi_i^\prime} |\sigma^\prime|\right) = \tilde O\left(n|T|+n|T^\prime|\right).
	\end{equation*}
		
	Let $S = S_{v(H\cap H^\prime)}$. For $d$ from $v(H\cap H^\prime)-1$ down to 0, for each depth-$d$ node $N$ in $S$, hash (\cref{hash}) the number of immediate subtrees of $N$ down from $\tilde O(n^{\psi_d})$ to $\tilde O(n^{\hat\phi_d})$, and if all of $N$'s children are null and $d>0$ then set $N$ to null. (We remark that $\hat\phi_d \le \psi_d$ by \cref{dimensions}.) This takes $\tilde O(|S|n) \le \tilde O((|T|+|T^\prime|)n) = \tilde O(n^{\max(\Delta(H), \Delta(H^\prime)) + 1})$ wires (\cref{tree-size}). By induction on $d$, a node retains its label if and only if it should retain its label in $T(H\cup H^\prime, \hat \pi)$, so the hashing succeeds w.h.p.\ because $X$ is good.
		
	Finally, for each leaf $(\tau,\tau^\prime)$ of $S$, append a copy of $\tau^\prime$ to each leaf of $\tau$, and put this in place of $(\tau,\tau^\prime)$ in $S$. This operation is purely semantic and requires no wires. The resulting tree does in fact have the proper number of children per node to be $T(H\cup H^\prime, (\pi^1,\dotsc,\pi^{v(H)},\pi^{\prime v(H\cap H^\prime)+1},\dotsc,\pi^{\prime v(H^\prime)}))$ by \cref{dimensions-2},\footnote{For $v(H\cap H^\prime) \le k < v(H)$ apply \cref{dimensions-2} with $L=H, R = H^\prime, A = H[\pi^1, \dotsc, \pi^k], B = H[\pi^1, \dotsc, \pi^{k+1}], C = A \cap B$, and for $v(H\cap H^\prime) \le k < v(H^\prime)$ apply \cref{dimensions-2} with $L = H^\prime, R = H, A = H^\prime[\pi^{\prime 1}, \dotsc, \pi^{\prime k}], B = H^\prime[\pi^{\prime 1}, \dotsc, \pi^{\prime k+1}], C = H$.} but without this knowledge we could instead use hashing on $\tau$ and $\tau^\prime$ as above, without knowing whether or not it succeeds vacuously.
\end{proof}

For each successive $H$ in an optimal union sequence, compute $T(H)$ as described above, and then apply a single $\textsf{OR}$ gate to all leaves of $T(G)$.

\section*{Acknowledgments}
\addcontentsline{toc}{section}{Acknowledgments}

Thanks to Benjamin Rossman for introducing me to this topic, and for having many helpful discussions about the research and about drafts of this paper. Thanks to Henry Yuen and the anonymous reviewers for their feedback as well. Part of this work was done while the author was visiting the Simons Institute for the Theory of Computing.

\appendix
\section{Equivalence of Threshold Weightings and Markov Chains}\label{markov}

\begin{thm}
	For any threshold weighting $(\alpha,\beta)\in\tcl(G)$ there exists a function $M:V(G)\times V(G)\rightarrow\R_{\ge0}$ such that
	\begin{enumerate}
		\item $M(u,u)=0$ for all $u$,
		\item $M(u,v)+M(v,u) = \beta(uv)$ for all $u\neq v$, and
		\item $\sum_{v\in V(G)} M(u,v) = \alpha(u)$ for all $u$.
	\end{enumerate}
\end{thm}
\begin{proof}
	Let $\Delta = (\alpha,\beta)$. The proof is by induction on $v(G)$. If $G$ is a single vertex $u$ then $\tcl(G)$ consists only of $\alpha=0$, so setting $M(u,u) = 0$ satisfies the requirements. Now assume $v(G) > 1$. For $A, B \subseteq G$ let $M(A, B) = \sum_{u \in V(A), v \in V(B)} M(u, v)$ (once $M(u,v)$ is specified). Assume without loss of generality that $G$ is a clique, since we can assign $\beta=0$ on nonexistent edges.
	
	Let $H = \mathrm{argmin}_{F\subset G, 0 < v(F) < v(G)}\Delta(F)$, where ties are broken arbitrarily subject to $H$ being an induced subgraph of $G$. Since $\Delta(G)=0$,
	\begin{equation*}
	\beta(H,G-H) = \Delta(G) + \beta(H,G-H) = \Delta(H) + \Delta(G-H) \ge \Delta(H),
	\end{equation*}
	so for $u\in V(H), v\in V(G-H)$ we can define $M(u,v)\in[0,\beta(uv)]$ such that $M(H,G-H)=\Delta(H)$. For $u\in V(H)$ let $\alpha_H(u) = \alpha(u) - M(u,G-H)$, and let $\Delta_H$ be the restriction of $\alpha_H - \beta$ to subgraphs of $H$. For any $\emptyset\subset F\subseteq H$,
	\begin{equation*}
	\Delta_H(F) = \Delta(F) - M(F,G-H) \ge \Delta(F) - M(H,G-H) \ge \Delta(H) - M(H,G-H) = 0,
	\end{equation*}
	with equality if $F=H$. Therefore $\Delta_H$ is a threshold weighting on $H$. Recursively define a restriction of $M$ to $V(H)\times V(H)$ such that this restriction is a Markov Chain on $H$ that is equivalent to $\Delta_H$.
	
	For $u \in V(G-H), v \in V(H)$ let $M(u,v) = \beta(uv) - M(v,u)$. For $u \in V(G-H)$ let $\alpha_{G-H}(u) = \alpha(u) - M(u,H)$, and let $\Delta_{G-H}$ be the restriction of $\alpha_{G-H} - \beta$ to subgraphs of $G-H$. Then,
	\begin{align*}
	\Delta_{G-H}(G-H) = \Delta(G-H) - M(G-H, H) = \Delta(G-H) - \beta(G-H, H) + M(H, G-H) \\
	= \Delta(G-H) - \beta(G-H, H) + \Delta(H) = \Delta(G) = 0.
	\end{align*}
	For any $\emptyset\subset F\subset G-H$, if $v(F) < v(G-H)$ then
	\begin{equation*}
	\Delta_{G-H}(F) = \Delta(F)-M(F,H) \ge \Delta(F)-\beta(F,H) \ge \Delta(G[V(H), V(F)])-\Delta(H) \ge 0,
	\end{equation*}
	and if $v(F) = v(G-H)$ then $\Delta_{G-H}(F) \ge \Delta_{G-H}(G-H) = 0$. Therefore $\Delta_{G-H}$ is a threshold weighting on $G-H$. Recursively define a restriction of $M$ to $V(G-H)\times V(G-H)$ such that this restriction is a Markov Chain on $G-H$ that is equivalent to $\Delta_{G-H}$.
	
	We now verify that $M(u,G) = \alpha(u)$ for all $u$; the other requirements follow easily by induction. If $u \in V(H)$ then $M(u,H) = \alpha_H(u)$ by induction, and $M(u,G-H) = \alpha(u) - \alpha_H(u)$ by the definition of $\alpha_H$. Similarly, if $u \in V(G-H)$ then $M(u,G-H) = \alpha_{G-H}(u)$ and $M(u,H) = \alpha(u) - \alpha_{G-H}(u)$. Therefore $M(u,G) = M(u,H) + M(u,G-H) = \alpha(u)$ for all $u$.
\end{proof}

\section{Proof that \texorpdfstring{$\kcol\left(\ham q d\right)$}{kappa(Kqd)} is \texorpdfstring{$O(q^d/d)$}{O(qd/d)} for all \texorpdfstring{$q$}{q}}\label{o-qdd-all}

The proof below is self-contained; however in places with clear analogues in \cref{o-qdd-even} we will give less detailed explanations of the intermediate steps and intuition.

Fix $q$ and $d$. Let a \emph{query tree} be a binary tree in which each node is labeled with some $U_1 \times \dotsb \times U_d$ where each $U_i \subseteq [q]$. The root is labeled with $[q]^d$, each leaf is labeled with a singleton set, and for any interior node $N$ labeled with $U_1 \times \dotsb \times U_d$ there exist $i \in [d]$ and $k \in U_i$ such that the left and right children of $N$ are labeled with $U_1 \times \dotsb \times U_{i-1} \times (U_i - k) \times U_{i+1} \times \dotsb \times U_d$ and $U_1 \times \dotsb \times U_{i-1} \times \{k\} \times U_{i+1} \times \dotsb \times U_d$ respectively. (In the latter case, $U_i$ necessarily has at least two elements.)
	
With respect to an implicit query tree $T$, let $\ell_0, \dotsc, \ell_{q^d-1}$ be the leaves in increasing order from left to right, and for $0 \le a \le q^d$ let $G(a) = \ham q d[\ell_0, \dotsc, \ell_{a-1}]$. Let $\mu_T = \max_a \Du(G(a))$ and let $\mu$ be the maximum of $\mu_T$ over all query trees $T$. For a threshold weighting $\Delta \in \tcl\left(\ham q d\right)$ and $H \subseteq \ham q d$ let $\kappa_\Delta(H) = \min_{S\in\seq(H)}\max_{F\in S}\Delta(F)$, and if $H$ is a single-vertex graph or the empty graph then let $\kappa_\Delta(H)=0$. By \cref{max-at-one} it suffices to prove that $\kappa_\Delta\left(\ham q d\right)$ is $O(q^d/d)$ for all threshold weightings $\Delta = (1,\beta)$.
	
\begin{lem}\label{some-lemma}
	Fix a query tree $T$. Let $0 \le a < b \le q^d$ such that $\ell_a, \dotsc, \ell_{b-1}$ are exactly the leaves descended from some node of $T$. Let $\Delta = (1,\beta) \in \tcl\left(\ham q d\right)$ such that $\beta(G(a)) \ge \Bu(G(a))$ and $\beta(G(b)) \ge \Bu(G(b))$, and $\kappa_\Delta(G(a)) \le 2\mu$. Then $\kappa_\Delta(G(b)) \le 2\mu$.
\end{lem}
\begin{proof}
	Let $N$ be the node of $T$ such that the leaves descended from $N$ are exactly $\ell_a, \dotsc, \ell_{b-1}$. Let $U_1 \times \dotsb \times U_d$ be the label of $N$. Let $B = G(b) - G(a) = \ham q d[U_1 \times \dotsb \times U_d] = \ham q d[\ell_a, \dotsc, \ell_{b-1}]$.
		
	The proof is by induction on $\sum_i (|U_i|-1)$, for all query trees $T$. (It follows from the definitions that $|U_i|\ge1$, with equality for all $i$ if and only if $N$ is a leaf.) In the inductive step we handle separately the cases where $\beta(B) \ge \Bu(B)$ and $\beta(B) < \Bu(B)$. The base case is a special case of the former because if $B$ is a single vertex then $\beta(B)$ and $\Bu(B)$ are both zero.
		
	\textbf{Case 1: $\beta(B) \ge \Bu(B)$.} If $B$ is a single vertex then $\kappa_\Delta(B) = 0 \le 2\mu$; we now obtain the same result in the case where $B$ is not a single vertex. Let $\mathcal I = \{i\in[d] \mid |U_i|\ge2\}$, and note that $\mathcal I$ is nonempty. For $i \in \mathcal I$ and $k \in U_i$ let $B(i,k) = B[v \in V(B) \mid v_i \neq k]$. Choose a pair $(\mathbf i, \mathbf k)$ uniformly at random out of all pairs $(i, k)$ such that $i \in \mathcal I$ and $k \in U_i$. Each edge in $B$ is also in $B(\mathbf i, \mathbf k)$ with the same probability $p = 1 - (|\mathcal I|+1)/\sum_{i\in\mathcal I}|U_i|$ (since adjacent vertices differ in a unique coordinate), so by linearity of expectation,
	\begin{equation*}
	\Ex[\beta(B(\mathbf i, \mathbf k))] = p \beta(B) \ge p \Bu(B) = \Ex[\Bu(B(\mathbf i, \mathbf k))].
	\end{equation*}
	Therefore $\beta(B(i,k)) \ge \Bu(B(i,k))$ for some fixed $i$ and $k$.
		
	Now our claim that $\kappa_\Delta(B) \le 2\mu$ follows from two applications of the inductive hypothesis. Let $T^\prime$ be any query tree in which the sequence of labels along the path from the root to the leftmost leaf includes $U_1 \times \dotsb \times U_d$ followed by $U_1 \times \dotsb \times (U_i - k) \times \dotsb \times U_d$. With respect to $T^\prime$, an application of the inductive hypothesis with $a^\prime = 0$ and $b^\prime = (1-1/|U_i|) \prod_j |U_j|$ reveals that $\kappa_\Delta(B(i,k)) \le 2\mu$, and then an application of the inductive hypothesis with $a^{\prime\prime} = (1-1/|U_i|) \prod_j |U_j|$ and $b^{\prime\prime} = \prod_j |U_j|$ ($= b-a$) reveals that $\kappa_\Delta(B) \le 2\mu$.
		
	The rest of the proof is essentially identical to the case $q = 2$. Let $S$ be an optimal (with respect to $\Delta$) union sequence for $G(a)$, followed by an optimal union sequence for $B$, followed by $G(a) \cup B, G(a)\cup B\cup e_1, \dotsc, G(a)\cup B\cup \{e_j\}$, where the $\{e_j\}$ are the edges between $G(a)$ and $B$ in $\ham q d$. (If $G(a)$ or $B$ lacks edges then omit certain graphs from this sequence.) Then,
	\begin{equation*}
	\max_{H\in S}\Delta(H) \le \max(\kappa_\Delta(G(a)), \kappa_\Delta(B), \Delta(G(a))+\Delta(B)).
	\end{equation*}
		
	We proceed to bound each of these three terms by $2\mu$, completing the proof. We have assumed that $\kappa_\Delta(G(a)) \le 2\mu$, and proved that $\kappa_\Delta(B) \le 2\mu$. We have also assumed that $\beta(G(a)) \ge \Bu(G(a))$, and since $\Delta$ and $\Du$ both evaluate to 1 on all vertices, it follows that $\Delta(G(a)) \le \Du(G(a)) \le \mu$ (with the last step following from the definition of $\mu$). Similarly, $\Delta(B) \le \Du(B) \le \mu$, and it follows that $\Delta(G(a)) + \Delta(B) \le 2\mu$.
		
	\textbf{Case 2: $\beta(B) < \Bu(B)$.} For $i \in \mathcal I$ and $k \in U_i$ let $H(i,k) = \ham q d[\ell_0, \dotsc, \ell_{a-1}, V(B(i,k))]$ (where $\mathcal I$ and $B(i,k)$ are defined as above). Choose a pair $(\mathbf i, \mathbf k)$ uniformly at random out of all pairs $(i, k)$ such that $i \in \mathcal I$ and $k \in U_i$. Then there exist $p_0 > p_1 > p_2 \ge 0$ (specifically, $p_0 = 1$, $p_1 = 1-|\mathcal I|/\sum_{i\in\mathcal I}|U_i|$, and $p_2 = 1 - (|\mathcal I|+1)/\sum_{i\in\mathcal I}|U_i|$) such that
	\begin{align*}
	\Ex[\beta(H(\mathbf i,\mathbf k))] &= p_0\beta(G(a)) + p_1\beta(G(a),B) + p_2\beta(B) \\
	&= (p_0-p_1)\beta(G(a)) + p_1\beta(G(b)) + (p_2-p_1)\beta(B) \\
	&> (p_0-p_1)\Bu(G(a)) + p_1\Bu(G(b)) + (p_2-p_1)\Bu(B) \\
	&= \Ex[\Bu(H(\mathbf i,\mathbf k))].
	\end{align*}
	Therefore $\beta(H(i,k)) > \Bu(H(i,k))$ for some fixed $i$ and $k$.
		
	Preparing to apply the inductive hypothesis, let $T^{\prime\prime}$ be any query tree structured and labeled exactly like $T$ on all ancestors of $\ell_j$ for all $j<a$, and on all ancestors of $N$, but now the left child of $N$ is labeled with $U_1 \times \dotsb \times (U_i - k) \times \dotsb \times U_d$. With respect to $T^{\prime\prime}$, an application of the inductive hypothesis with $a^\prime = a$ and $b^\prime = a + (1-1/|U_i|) \prod_j |U_j|$ reveals that $\kappa_\Delta(G(a + (1-1/|U_i|) \prod_j |U_j|)) \le 2\mu$, and a second application of the inductive hypothesis with $a^{\prime\prime} = a + (1-1/|U_i|) \prod_j |U_j|$ and $b^{\prime\prime} = b$ ($= a + \prod_j|U_j|$) reveals that $\kappa_\Delta(G(b)) \le 2\mu$.
\end{proof}

\begin{lem}\label{some-other-lemma}
	$\mu$ is $O(q^d/d)$.
\end{lem}
\begin{proof}
	We use a cruder bound here than in the case $q=2$. Let $T$ be an arbitrary query tree and $a \in [q^d]$. Let $N_0$ be the nearest common ancestor of $\ell_0$ and $\ell_{a-1}$. If $N_0$ is a leaf then clearly $\Du(G(a))$ is $O(q^d/d)$, so assume otherwise. Let $N_L$ and $N_R$ be the left and right children of $N_0$, and note that $\ell_{a-1}$ is a descendant of $N_R$. By \cref{mc-eq}, since $\ham q d$ is $(q-1)d$-regular it suffices to prove that $e(G(a), \ham q d - G(a))$ is $O(q^{d+1})$. Suppose $N_0$ is labeled with $U_1 \times \dotsb \times U_d$ and $N_L$ is labeled with $U_1 \times \dotsb \times (U_i - k) \times \dotsb \times U_d$. Since each vertex in $G(a)$ is a descendant of $N_0$, all edges between $G(a)$ and $\ham q d - G(a)$ are in one of the following classes:
	\begin{enumerate}
		\item \emph{Edges (in $\ham q d$) between a leaf descended from $N_0$ and a leaf not descended from $N_0$.} Each leaf descended from $N_0$ is adjacent to $\sum_{j=1}^d (q-|U_j|)$ leaves not descended from $N_0$, so this amounts to $(dq - \sum_j |U_j|) \prod_j |U_j|$ edges in total. By the AM-GM inequality, this is at most
		\begin{equation*}
		\left(\frac{(dq-\sum_j |U_j|) + \sum_j |U_j|}{d+1}\right)^{d+1} = \left(\frac{dq}{d+1}\right)^{d+1} = q^{d+1}\left(1 - \frac1{d+1}\right)^{d+1} < q^{d+1}/e.
		\end{equation*}
		\item \emph{Edges (in $\ham q d$) between a leaf descended from $N_L$ and a leaf descended from $N_R$.} Each leaf descended from $N_L$ is adjacent to one leaf descended from $N_R$, so this amounts to at most $\prod_j |U_j| \le q^d$ edges in total.
		\item \emph{Edges (in $\ham q d$) between a leaf descended from $N_R$ that's in $G(a)$, and a leaf descended from $N_R$ that's in $\ham q d - G(a)$.} This is at most what the value of $\mu$ would be if $d$ were $d-1$ instead. (Eliminate coordinate $i$, and replace $U_j$ with $[q]$ for all $j\neq i$.)
	\end{enumerate}
	The total number of edges in all classes is therefore $O(q^{d+1} + q^d + \dotsb) = O(q^{d+1})$.
\end{proof}

Finally, it follows from \cref{some-lemma,some-other-lemma} that $\kcol\left(\ham q d\right) \le 2\mu \le O(q^d/d)$.

\begin{remark*}
	The above argument holds even if we relax the definition of threshold weightings to allow $\Delta$ to take on negative values (where all definitions in terms of threshold weightings are with respect to this revised definition).
\end{remark*}

\section{Properties of Threshold Weightings and Threshold Random Graphs}\label{ptw}

\begin{lem}\label{janson-app}
	If $A \subseteq U \subseteq G$ are fixed graphs, $\Delta \in \tcl(G)$, and $\Delta(A) < \Delta(H)$ for all $H \in (A, U]$, then conditional on $\mathcal A \in \colt A {\rg X \Delta n}$, there are at least $n^{\Delta(U) - \Delta(A)}(1-o(1))$ $U$-extensions of $\mathcal A$ a.a.s.
\end{lem}

Li et al.~\cite{LRR17} stated without proof that a similar result can be obtained using Janson's Inequality~\cite{Jan90}:
\begin{fact}[Janson's Inequality]\label{janson}
	Let $\mathbf B_1, \ldots, \mathbf B_\ell$ be independent Bernoulli random variables, let $W_1, \ldots, W_k \subseteq [\ell]$, and for $i\in[k]$ let $\mathbf I_i = \prod_{j\in W_i} \mathbf B_j$. Also for $i,j\in[k], i\neq j$ let $i\sim j$ if $W_i \cap W_j \neq \emptyset$. Let $\mathbf S = \sum_i \mathbf I_i$ and $\mu = \Ex[\mathbf S]$. Then for all $0 \le \epsilon \le 1$,
	\begin{equation*}
		P(\mathbf S \le (1-\epsilon)\mu) \le \exp\left(-\frac{\epsilon^2}2 \cdot \frac{\mu^2}{\mu + \sum_{i \sim j}\Ex[\mathbf I_i\mathbf I_j]}\right).
	\end{equation*}
\end{fact}

\begin{proof}[Proof of \cref{janson-app}]
	Let $\mathcal U_1, \dotsc, \mathcal U_k$ be the possible $U$-extensions of $\mathcal A$, and let $\mathbf I_i = \Ind{\mathcal U_i \subseteq \mathbf X}$. Define $\mu$ as in \cref{janson}; clearly $\mu = n^{\Delta(U) - \Delta(A)}$, by reasoning similar to the proof of \cref{ex-sub-count}. If $i \sim j$ then the projection of $\mathcal U_i \cap \mathcal U_j$ onto $U$ must be some graph in $(A, U)$, so
	\begin{equation*}
		\sum_{i \sim j}\Ex[\mathbf I_i \mathbf I_j] \le \sum_{\mathclap{H \in (A, U)}} \mu n^{\Delta(U)-\Delta(H)} = \mu^2 \sum_{\mathclap{H \in (A, U)}} n^{\Delta(A)-\Delta(H)} = o(\mu^2).
	\end{equation*}
	Since $\mu$ is also $o(\mu^2)$, it follows that $\mu^2\big/\left(\mu + \sum_{i \sim j}\Ex[\mathbf I_i\mathbf I_j]\right) \ge \mu^2/o(\mu^2) = \omega(1)$, and the result follows from \cref{janson}.
\end{proof}

\begin{lem}\label{dimensions}
	For all $A\subseteq B\subseteq F\subseteq H\subseteq G$ and $\Delta \in \tcl(G)$,
	\begin{equation*}
	\dst_H(B)-\dst_H(A) \le \dst_F(B)-\dst_F(A).
	\end{equation*}
\end{lem}
\begin{proof}
	Since $B \subseteq \Gamma_F(B) \subseteq \Gamma_H(A) \cup \Gamma_F(B)$ it follows that $\dst_H(B) \le \Delta(\Gamma_H(A)\cup\Gamma_F(B))$, and since $A \subseteq \Gamma_H(A)$ and $A\subseteq B\subseteq \Gamma_F(B)$ it follows that $\dst_F(A) \le \Delta(\Gamma_H(A)\cap\Gamma_F(B))$. So by \cref{Delta-sum,Gamma-lemma-1},
	\begin{align*}
	\dst_H(B) + \dst_F(A) &\le \Delta(\Gamma_H(A) \cup \Gamma_F(B)) + \Delta(\Gamma_H(A) \cap \Gamma_F(B)) \\
	&= \Delta(\Gamma_H(A)) + \Delta(\Gamma_F(B)) \\
	&= \dst_H(A) + \dst_F(B).\qedhere
	\end{align*}
\end{proof}

\begin{lem}\label{dimensions-2}
	Let $\Delta \in \tcl(G)$ and assume $L \cap R \subseteq A \subseteq B \subseteq L \subseteq G$ and $L\cap R \subseteq C \subseteq R \subseteq G$. Then, $\dst_{L\cup R}(B\cup C)-\dst_{L\cup R}(A\cup C) = \dst_L(B)-\dst_L(A)$.
\end{lem}
\begin{proof}
	For all $F\in[A,L]$ and $H\in[C,R]$,
	\begin{align*}
	F\cap H &\subseteq L\cap R &&\text{($F\subseteq L$ and $H\subseteq R$)}\\
	&\subseteq A\cap C &&\text{(by assumption)}\\
	&\subseteq F\cap H, &&\text{($A\subseteq F$ and $C\subseteq H$)},
	\end{align*}
	so by \cref{Delta-sum},
	\begin{align*}
	\dst_{L\cup R}(A\cup C) &= \min_{\mathclap{\substack{A\subseteq F\subseteq L \\ C\subseteq H\subseteq R}}}\Delta(F\cup H)
	= \min_{\mathclap{\substack{A\subseteq F\subseteq L \\ C\subseteq H\subseteq R}}}(\Delta(F)+\Delta(H)-\Delta(F\cap H)) \\
	&= \min_{\mathclap{A\subseteq F\subseteq L}}\Delta(F) + \min_{\mathclap{C\subseteq H\subseteq R}}\Delta(H)-\Delta(L\cap R)
	= \dst_L(A) + \dst_R(C)-\Delta(L\cap R).
	\end{align*}
	The same reasoning applies with $B$ in place of $A$, so
	\begin{equation*}
	\dst_{L\cup R}(A\cup C) - \dst_L(A) = \dst_R(C) - \Delta(L\cap R) = \dst_{L\cup R}(B\cup C) - \dst_L(B). \qedhere
	\end{equation*}
\end{proof}

\printbibliography[heading=bibintoc]

\end{document}